\documentclass[11pt]{article} 

\usepackage[paper=letterpaper, margin=1in]{geometry}

\usepackage{ifthen}

\ifthenelse{\isundefined{\GenerateShortVersion}}
{
\def\LongVersion{}
\def\LongVersionEnd{}
\long\def\ShortVersion#1\ShortVersionEnd{}
}
{
\def\ShortVersion{}
\def\ShortVersionEnd{}
\long\def\LongVersion#1\LongVersionEnd{}
}

\usepackage{ifthen}

\newcommand{\Ignore}[1]{\ignorespaces}

\usepackage{amsmath}
\usepackage{amssymb}

\usepackage{amsthm}

\usepackage{ifpdf}

\usepackage{authblk}

\ifpdf
\usepackage[pdftex]{graphicx}
\else
\usepackage[dvips]{graphicx}
\fi

\usepackage[%
  ocgcolorlinks,%
  linkcolor=blue,%
  filecolor=blue,%
  citecolor=blue,%
  urlcolor=blue]{hyperref}

\usepackage{enumitem}

\usepackage{multirow}
\usepackage[singlelinecheck=false, labelfont=bf]{caption}

\usepackage{cleveref}
\usepackage{csquotes}

\usepackage[noend,ruled,linesnumbered]{algorithm2e}

\usepackage[T1]{fontenc}

\LongVersion 
\LongVersionEnd 

\ShortVersion 
\usepackage{times}
\ShortVersionEnd 

\ShortVersion 
\renewcommand{\paragraph}[1]{\par\noindent\textbf{#1}}
\ShortVersionEnd 

\newtheorem{theorem}{Theorem}[section]
\newtheorem{lemma}[theorem]{Lemma}

\newtheorem{corollary}[theorem]{Corollary}

\newtheorem{claim}[theorem]{Claim}

\newtheorem*{corollary*}{Corollary}

\theoremstyle{definition}

\newtheorem{remark}[theorem]{Remark}

\newtheorem*{definition*}{Definition}
\theoremstyle{plain}

\newenvironment{subproof}[1][\proofname]{%
\begin{proof}[#1]%
}{%
\end{proof}%
}

\def\qedlabel#1{\def\theqedlabel{#1}}

\newcommand{\NextTime}[2]{{A}_{#1}(#2)}

\newcommand{\predictJ}[2]{a_{#1}^{#2}}
\newcommand{\predict}[1]{a_{#1}}
\newcommand{\predictHat}[2]{\hat{a}_{#1}(#2)}
\newcommand{\LastRequest}[2]{R_{#1}(#2)}
\newcommand{\NULL}{-1}

\newcommand{\HD}[2]{\mathcal{H}_{#1}^{#2}}

\newcommand{\ActI}{\theta}
\newcommand{\ActL}{\mathcal{L}}

\newcommand{\Alg}{\ensuremath{\mathtt{Alg}}}
\newcommand{\OPT}{\ensuremath{\mathtt{OPT}}}
\newcommand{\cost}{\operatorname{cost}}
\newcommand{\regret}{\operatorname{regret}}

\newcommand{\INV}{\mathtt{INV}}

\newcommand{\expert}{predictor}
\newcommand{\experts}{predictors}
\newcommand{\Expert}{Predictor}
\newcommand{\Experts}{Predictors}

\newcommand{\FitF}{\ensuremath{\mathtt{FitF}}}
\newcommand{\Sim}{\ensuremath{\mathtt{Sim}}}

\newcommand{\ErrorRounds}{\sharp \mathtt{ErrorRounds}}
\newcommand{\LpNorm}{\sharp \mathtt{L}_{1}}
\newcommand{\InvertedPairs}{\sharp \mathtt{InvertedPairs}}
\newcommand{\InvertedRounds}{\sharp \mathtt{InvertedRounds}}
\newcommand{\RoundsWithErrorAndInversion}{\sharp \mathtt{ErrorRoundsInInversion}}
\newcommand{\StateDistance}{\sharp \mathtt{Distances}}

\newcommand{\Ex}{\mathbb{E}}

\newcommand{\Evictions}{\sharp \mathtt{Evictions}}

\newcommand{\SightlessChasing}{\texttt{S-C\&S}}

\newcommand{\FundingEmek}{The work of Yuval Emek was supported in part by an Israeli Science Foundation grant number 1016/17.}
\newcommand{\FundingKutten}{The work of Shay Kutten was supported in part by a grant from the ministry of science in the program that is joint with JSPS and in part by the BSF.}
\newcommand{\FundingShi}{The work of Yangguang Shi was partially supported at the Technion by a fellowship of the Israel Council for Higher Education.}

\title{Online Paging with a Vanishing Regret}
\author{Yuval Emek,\thanks{\FundingEmek{}} $\;$ Shay Kutten,\thanks{\FundingKutten{}} $\;$ Yangguang Shi\thanks{\FundingShi{}}}
\affil{Faculty of Industrial Engineering and Management, Technion, Haifa, Israel\\
\texttt{\{yemek, kutten\}@technion.ac.il}, \texttt{shiyangguang@campus.technion.ac.il}}

\date{}

\begin{document}

\begin{titlepage}

\maketitle

\begin{abstract}
This paper considers a variant of the online \emph{paging} problem, where the
online algorithm has access to multiple \emph{\experts{}}, each producing a
sequence of predictions for the page arrival times.
The \experts{} may have occasional prediction errors and it is assumed
that at least one of them makes a sublinear number of prediction errors in
total.
Our main result states that this assumption suffices for the design of a
randomized online algorithm whose time-average \emph{regret} with respect to
the optimal offline algorithm tends to zero as the time tends to infinity.
This holds (with different regret bounds) for both the \emph{full information}
access model, where in each round, the online algorithm gets the predictions
of all \experts{}, and the \emph{bandit} access model, where in each round,
the online algorithm queries a single \expert{}.

While online algorithms that exploit inaccurate predictions have been a topic of growing interest in the last few years, to the best of our knowledge, this is the first paper that studies this topic in the context of multiple \experts{} for an online problem with unbounded request sequences.
Moreover, to the best of our knowledge, this is also the first paper that aims
for (and achieves) online algorithms with a vanishing regret for a classic online problem under reasonable assumptions.
\end{abstract}


\thispagestyle{empty}

\end{titlepage}


\section{Introduction} \label{section:intro}
A critical bottleneck in the performance of digital computers, known as the
``memory wall'', is that the main memory (a.k.a.\ DRAM) is several orders of
magnitude slower than the multiprocessor
\cite{WulfM95, HashemiSSALCKR18, AyersLKR20}.
Modern computer architectures bridge this performance gap by utilizing a
cache, namely, a memory structure positioned next to the multiprocessor that
responds much faster than the main memory.
However, the cache is inherently smaller than the main memory which means that
some of the memory items requested by the running program may be missing from
the cache.
When such a \emph{cache miss} occurs, the multiprocessor is required to fetch
the requested item from the main memory into the cache;
if the cache is already full, then some previously stored item must be evicted
to make room for the new one.
Minimizing the number of cache misses is known to be a primary criterion for
improving the computer's performance
\cite{WulfM95, abs-2006-16239}.

The aforementioned challenge is formalized by means of a classic \emph{online}
problem called \emph{paging} \cite{Sleator1985AELUPR} (a.k.a.\ unweighted
\emph{caching}), defined over a main memory that consists of
$n \in \mathbb{Z}_{> 0}$
\emph{pages} and a cache that holds
$k \in \mathbb{Z}_{> 0}$
pages at any given time,
$k < n$.
The execution of a paging algorithm \Alg{} progresses in
$T \in \mathbb{Z}_{> 0}$
discrete \emph{rounds}, where round
$t \in T$
occupies the time interval
$[t, t - 1)$.
An instance of the paging problem is given by a sequence
$\sigma = \{ \sigma_{t} \}_{t \in [T]}$
of page \emph{requests} so that request
$\sigma_{t} \in [n]$
is revealed at time
$t \in [T]$.
Denoting the \emph{cache configuration} of \Alg{} at time $t$ by
$C_{t} \subset [n]$,
$|C_{t}| = k$,
if
$\sigma_{t} \in C_{t}$,
then \Alg{} does nothing in round $t$;
otherwise
($\sigma_{t} \notin C_{t}$),
a \emph{cache miss} occurs and \Alg{} should bring the requested page into the
cache so that
$\sigma_{t} \in C_{t + 1}$.
Since
$|C_{t + 1}| = |C_{t}| = k$,
it follows that upon a cache miss, \Alg{} must evict some page
$i \in C_{t}$
and its policy is reduced to the selection of this page $i$.
The \emph{cost} incurred by \Alg{} on $\sigma$ is defined to be the number of
cache misses it suffers throughout the execution, denoted by
\[
\cost_{\sigma}(\Alg)
\, = \,
\left|\left\{ t \in [T] : \sigma_{t} \notin C_{t} \right\}\right| \, ,
\]
taking the expectation if \Alg{} is a randomized algorithm.
When $\sigma$ is clear from the context, we may omit the subscript, writing
$\cost(\Alg) = \cost_{\sigma}(\Alg)$.

\subparagraph*{NAT and the \FitF{} Algorithm.}
To avoid cumbersome notation, we assume hereafter that the request sequence
$\sigma$ is augmented with a suffix of $n$ virtual requests so that
$\sigma_{T + i} = i$
for every
$i \in [n]$.
This facilitates the definition of the \emph{next arrival time (NAT)} of page
$i \in [n]$
with respect to time
$t \in [T]$
as the first time after $t$ at which page $i$ is requested, denoted by
\[
\NextTime{t}{i}
\, = \,
\min \{ t' > t \mid \sigma_{t'} = i \} \, .
\]
Based on that, we can define the \FitF{} (stands for furthest in the future)
paging algorithm that on a cache miss at time
$t \in [T]$,
evicts the page
$i \in C_{t}$
that maximizes $\NextTime{t}{i}$.
A classic result of Belady \cite{Belady1966study} states that \FitF{} is
optimal in terms of the cost it incurs for the given request sequence
$\sigma$;
we subsequently denote
$\OPT_{\sigma} = \cost_{\sigma}(\FitF)$
and omit the subscript, writing
$\OPT = \OPT_{\sigma}$,
when $\sigma$ is clear from the context.
It is important to point out that \FitF{} is an \emph{offline} algorithm as online algorithms are oblivious to the NATs.

\subparagraph*{Regret.}
We define the \emph{regret} of an online paging algorithm \Alg{} on $\sigma$ as
\[
\regret_{\sigma}(\Alg)
\, = \,
\cost_{\sigma}(\Alg) - \OPT_{\sigma}
\]
and omit the subscript, writing
$\regret(\Alg) = \regret_{\sigma}(\Alg)$
when $\sigma$ is clear from the context.
Our goal in this paper is to develop an online algorithm that admits a
\emph{vanishing regret}, namely, an online algorithm \Alg{} for which it is
guaranteed that
\[
\lim_{T \rightarrow \infty}
\frac{\sup \left\{ \regret_{\sigma}(\Alg) \mid \sigma \in [n]^{T} \right\}}{T}
\, = \,
0 \, .
\]
The following theorem states that this goal is hopeless unless the online
algorithm has access to some additional information;
its proof should be a folklore, we add it in Appendix~\ref{section:impossibility} for completeness.

\begin{theorem} \label{theorem:logarithmic-gap-with-no-additional-info}
Fix
$n = k + 1$
and let $\sigma$ be a request sequence generated by picking $\sigma_{t}$
uniformly at random (and independently) from $[n]$ for
$t = 1, \dots, T$.
Then,
$\Ex \left( \cost(\Alg) \right) \geq \Omega \left( \frac{T}{k} \right)$
for any (possibly randomized) online paging algorithm \Alg{}, whereas
$\Ex(\OPT) \leq O \left( \frac{T}{k \log k} \right)$.
\end{theorem}

\subsection{Machine Learned Predictions}
\label{section:machine-learning}
Developments in \emph{machine learning (ML)} technology suggest a new
direction for reducing the number of cache misses by means of predicting the
request sequence.
Indeed, recent studies have shown that neural networks can be employed to
predict the memory pages accessed by a program with high accuracy
\cite{HashemiSSALCKR18, braun2019understanding, SrivastavaLBKP19, PeledWE20,
SrivastavaWZRKP20}.
When provided with an accurate prediction of the request sequence $\sigma$,
one can simply simulate \FitF{}, thus ensuring an optimal performance.

Unfortunately, the predictions generated by ML techniques are usually not
100\% accurate as a result of a distribution drift between the training and
test examples or due to adversarial examples
\cite{SzegedyZSBEGF13, Lykouris2018competitive}.
This gives rise to a growing interest in developing algorithmic techniques
that can overcome inaccurate predictions, aiming for the design of online
algorithms with performance guarantee that improves as the predictions become
more accurate
\cite{Lykouris2018competitive, Rohatgi2020near, Antoniadis2020online,
Wei2020better}.
The existing literature in this line of research studies a setting where the
online algorithm \Alg{} is provided with a sequence of predictions for
$\sigma$ and focuses on bounding \Alg{}'s \emph{competitive ratio} as a
function of the proximity of this sequence to $\sigma$ (more on that in
Section~\ref{section:comparison}).

The current paper tackles the challenge of overcoming inaccurate predictions
from a different angle:
Motivated by the abundance of forecasting algorithms that may be trained on
different data sets or using different models (e.g., models that are robust to
adversarial examples \cite{GoodfellowSS14}), we consider a decision maker with
access to \emph{multiple} predicting sequences for $\sigma$.
Our main goal is to design an online algorithm \Alg{} that admits a vanishing
regret assuming that at least one of the predicting sequences is sufficiently
accurate, even though the decision maker does not know in advance which
predicting sequence it is.

\subparagraph*{Explicit \Experts{}.}
Formally, we consider
$M \in \mathbb{Z}_{> 0}$
\emph{\experts{}} whose role is to predict the request sequence $\sigma$.
In the most basic form, referred to hereafter as the \emph{explicit
\experts{}} setting, each \expert{}
$j \in [M]$
produces a page sequence
$\pi^{j} = \{ \pi^{j}_{t} \}_{t \in [T]} \in [n]^{T}$,
where $\pi^{j}_{t}$ aims to predict $\sigma_{t}$ for every
$t \in [T]$,
and the sequences
$\pi^{1}, \dots, \pi^{M}$
are revealed to the online algorithm \Alg{} at the beginning of the execution.
Under the explicit \experts{} setting, \expert{}
$j \in [M]$
is said to have a \emph{prediction error} in round
$t \in [T]$
if
$\pi^{j}_{t} \neq \sigma_{t}$.
We measure the accuracy of \expert{} $j$ by means of her \emph{cumulative
prediction error}
\[
\eta_{e}^{j}
\, = \,
\eta_{e}(\pi^{j})
\, = \,
\left| \left\{ t \in [T] : \pi^{j}_{t} \neq \sigma_{t} \right\} \right|
\]
and define
$\eta_{e}^{\min} = \min \{ \eta_{e}^{j} \mid j \in [M] \}$.

The fundamental assumption that guides the current paper, referred to
hereafter as the \emph{good \expert{}} assumption, is that there exists at
least one \expert{} whose cumulative prediction error is sublinear in $T$,
namely,
$\eta_{e}^{\min} = o (T)$.
We emphasize that \Alg{} has no a priori knowledge of
$\eta_{e}^{1}, \dots, \eta_{e}^{M}$
nor does it know the \expert{} that realizes $\eta_{e}^{\min}$.
Our main research question can now be stated as follows:

\begin{displayquote}
Does the good \expert{} assumption provide a sufficient condition for the
existence of an online algorithm that admits a vanishing regret?
\end{displayquote}

\subparagraph*{NAT \Experts{}.}
For the paging problem, it is arguably more natural to consider the setting of
\emph{NAT \experts{}}, where \expert{}
$j \in [M]$
produces in each round
$t \in [T]$,
a prediction
$\predictJ{t}{j} \in (t, T + n]$
for the NAT $\NextTime{t}{\sigma_{t}}$ of the page that has just been
requested.
Under this setting, \expert{}
$j \in [M]$
is said to have a \emph{prediction error} in round
$t \in [T]$
if
$\predictJ{t}{j} \neq \NextTime{t}{\sigma_{t}}$.
As in the explicit predictors setting, we measure the accuracy of (NAT)
\expert{} $j$ by means of her \emph{cumulative prediction error}, now defined
as
\begin{equation} \label{equation:nat-cumulative-prediction-error}
\eta_{N}^{j}
\, = \,
\left| \left\{ t \in [T] : \predictJ{t}{j} \neq \NextTime{t}{\sigma_{t}}
\right\} \right|
\end{equation}
(this measure is termed \emph{classification loss} in
\cite{Lykouris2018competitive}), and define
$\eta_{N}^{\min} = \min \{ \eta_{N}^{j} \mid j \in [M] \}$.\footnote{%
In \Cref{section:comparison}, we provide a refined definition for the cumulative prediction error of a NAT predictor that is more robust against adversarial interference such as shifting each $\predictJ{j}{t}$ by a constant.
For simplicity of the exposition, the definition presented in Eq.~\eqref{equation:nat-cumulative-prediction-error} is used throughout the current section;
we emphasize though that all our results hold for the stronger notion of prediction error as defined in \Cref{section:comparison}.}
The NAT predictors version of the good \expert{} assumption states that
$\eta_{N}^{\min} = o (T)$.

Given a page sequence
$\pi = \{ \pi_{t} \}_{t \in [T]} \in [n]^{T}$
augmented with a suffix of $n$ pages such that
$\pi_{T + i} = i$
for every
$i \in [n]$,
we say that (NAT) \expert{}
$j \in [M]$
is \emph{consistent} with $\pi$ if
$\predictJ{t}{j}
=
\min \{ t' > t \mid \pi_{t'} = \sigma_{t} \}$
for every
$t \in [T]$;
if the page sequence $\pi$ is not important or clear from the context, then we
may say that \expert{} $j$ is \emph{consistent} without mentioning $\pi$.
The key observation here is that if \expert{} $j$ is consistent with a page
sequence $\pi$, then $\eta_{N}^{j}$ provides a good approximation for
$\eta_{e}(\pi)$, specifically,
\begin{equation} \label{equation:full-information-vs-explicit}
\eta_{e}(\pi) - n
\, \leq \,
\eta_{N}^{j}
\, \leq \,
2 \cdot \eta_{e}(\pi)
\end{equation}
(see Lemma~\ref{lemma_errors_transition}).
This means that the setting of NAT predictors is stronger than that of
explicit predictors in the sense that NAT \expert{}
$j \in [M]$
can be simulated (consistently) from explicit \expert{} $j$ by deriving the
NAT prediction $\predictJ{t}{j}$ in round
$t \in [T]$
from the (explicit) predictions
$\pi^{j}_{t + 1}, \pi^{j}_{t + 2}, \dots, \pi^{j}_{T}$,
while ensuring that $\eta_{N}^{j}$ is a good approximation for $\eta_{e}^{j}$.
Therefore, unless stated otherwise, we subsequently restrict our attention to
NAT predictors and in particular omit the subscript from the cumulative
prediction error notation, writing
$\eta^{j} = \eta_{N}^{j}$
and
$\eta^{\min} = \eta_{N}^{\min}$.
It is important to point out though that the results established in the
current paper hold regardless of whether the (NAT) predictors are consistent
or not.

\subparagraph*{Access Models.}
Recall that the (NAT) \experts{}
$j \in [M]$
produce their predictions in an online fashion so that the NAT prediction
$\predictJ{t}{j}$ is produced in round $t$.
This calls for a distinction between two \emph{access models} that determine
the exact manner in which $\predictJ{t}{j}$ is revealed to the online paging
algorithm \Alg{}.
First, we consider the \emph{full information} access model, where in each
round
$t \in [T]$,
\Alg{} receives $\predictJ{t}{j}$ for all
$j \in [M]$.
Motivated by systems in which accessing the ML predictions is costly in both
time and space (thus preventing \Alg{} from querying multiple \experts{} in
the same round and/or predictions belonging to past rounds), we also consider
the \emph{bandit} access model, where in each round
$t \in [T]$,
\Alg{} receives $\predictJ{t}{j}$ for a single \expert{}
$j \in [M]$
selected by \Alg{} in that round.
To make things precise, we assume, under both access models, that if \Alg{}
has to evict a page in round $t$, then the decision on the evicted page is
made prior to receiving the prediction(s) in that round.
Notice though that the information that \Alg{} receives from the \expert{}(s)
is not related to the evicted page and as such, should not be viewed as a
feedback that \Alg{} receives in response to the action it takes in the
current round.

\subsection{Our Contribution}
\label{section:contribution}
Consider a (single) \expert{} that in each round
$t \in [T]$,
produces a prediction $\predict{t}$ for the NAT $\NextTime{t}{\sigma_{t}}$ of
the page that has just been requested and let $\eta$ be her cumulative
prediction error.
Our first technical contribution comes in the form pf a thorough analysis of
the performance of a simple online paging algorithm called \Sim{} that
simulates \FitF{}, replacing the actual NATs with the ones derived from the
prediction sequence
$\{ \predict{t} \}_{t \in [T]}$.
Using some careful combinatorial arguments, we establish the following bound.

\begin{theorem} \label{theorem:bound-regret-Sim}
The regret of \Sim{} satisfies
$\regret(\Sim) \leq O (\eta + k)$.
\end{theorem}

Relying on online learning techniques, Blum and Burch \cite{Blum2000online}
develop an online algorithm ``multiplexer'' that given multiple online
algorithms as subroutines, produces a randomized online algorithm that
performs almost as good as the best subroutine in hindsight (see
Theorem~\ref{theorem_Blum_online_learning}).
Applying Theorem~\ref{theorem:bound-regret-Sim} to the $M$ \experts{} so that
each \expert{}
$j \in [M]$
yields its own online paging algorithm $\Sim^{j}$ and plugging algorithms
$\Sim^{1}, \dots, \Sim^{M}$
into the multiplexer of \cite{Blum2000online}, we establish the following theorem, thus concluding that the good \expert{} assumption implies an online paging algorithm with a vanishing regret under the full information access model.

\begin{theorem} \label{theorem:bound-regret-full-information-access}
There exists a randomized online paging algorithm that given full information
access to $M$ NAT \experts{} with minimum cumulative prediction error
$\eta^{\min}$,
has regret at most
$O \left( \eta^{\min} + k + \left( T k \log{M} \right)^{1 / 2} \right)$.
\end{theorem}

Combined with (\ref{equation:full-information-vs-explicit}), we obtain the
same asymptotic regret bound for explicit \experts{}.

\sloppy
\begin{corollary} \label{corollary:bound-regret-explicit-experts}
There exists a randomized online paging algorithm that given access to $M$
explicit \experts{} with minimum cumulative prediction error
$\eta_{e}^{\min}$,
has regret at most
$O \left( \eta_{e}^{\min} + k + \left( T k \log{M} \right)^{1 / 2} \right)$.
\end{corollary}
\par\fussy

The explicit \experts{} setting is general enough to make it applicable to
virtually any online problem.
This raises the question of whether other online problems admit online
algorithms with a vanishing regret given access to explicit \experts{} whose
minimum cumulative prediction error is sublinear in $T$.
We view the investigation of this question as an interesting research thread
that will hopefully arise from the current paper.

Going back to the setting of NAT \experts{}, one wonders if a vanishing regret
can be achieved also under the bandit access model since the technique of
\cite{Blum2000online} unfortunately does not apply to this more restricted
access model.
An inherent difficulty in the bandit access model is that we cannot keep
track of the cache configuration of $\Sim^{j}$ unless \expert{} $j$ is queried
in each round (which means that no other \expert{} can be queried).
To overcome this obstacle, we exploit certain combinatorial properties of the \Sim{} algorithm to show that $\Sim^{j}$ can be ``chased'' without knowing its current cache
configuration, while bounding the accumulated cost difference.
By a careful application of online learning techniques, this allows us to
establish the following theorem, thus concluding that the good \expert{}
assumption implies an online paging algorithm with a vanishing regret
under the bandit access model as well.

\sloppy
\begin{theorem} \label{theorem:bound-regret-bandit-access}
There exists a randomized online paging algorithm that given bandit access to
$M$ NAT \experts{} with minimum cumulative prediction error $\eta^{\min}$, has
regret at most
$O \left( \eta^{\min} + T^{2 / 3} k M^{1 / 2} \right)$.
\end{theorem}
\par\fussy

\subsection{Related Work and Discussion}
\label{section:related-work}
We say that an online algorithm \Alg{} for a minimization problem
$\mathcal{P}$ has \emph{competitive ratio} $\alpha$ if for any instance
$\sigma$ of $\mathcal{P}$, the cost incurred by \Alg{} on $\sigma$ is at most
$\alpha \cdot \OPT_{\sigma} + \beta$,
where $\OPT_{\sigma}$ is the cost incurred by an optimal offline algorithm on
$\sigma$ and $\beta$ is a constant that may depend on $\mathcal{P}$, but not
on $\sigma$
\cite{Sleator1985AELUPR, Borodin1998online}.
In comparison, the notion of regret as defined in the current paper uses the
optimal offline algorithm as an absolute (additive), rather than relative
(multiplicative), benchmark.
Notice that the vanishing regret condition cannot be expressed in the scope of
the competitive ratio definition.
In particular,
$\alpha = 1$
is a stronger requirement than vanishing regret as the latter can accommodate
an additive parameter $\beta$ that does depend on $\sigma$ as long as it is
sublinear in
$T = |\sigma|$.
On the other hand,
$\alpha > 1$
implies a non-vanishing regret when $\OPT_{\sigma}$ scales linearly with
$T$.

As mentioned in Section~\ref{section:machine-learning}, most of the existing
literature on augmenting online algorithms with ML predictions is restricted to the case of a single \expert{}
\cite{Lykouris2018competitive, Rohatgi2020near, Antoniadis2020online,
Wei2020better}.
The goal of these papers is to develop online algorithms with two guarantees:
(i)
their competitive ratio tends to
$O (1)$
(though not necessarily to $1$) as the \expert{}'s accuracy improves;
and
(ii)
they are robust in the sense that regardless of the \expert{}'s accuracy,
their competitive ratio is not much worse than that of the best online
algorithm that has no access to predictions.

In contrast, the current paper addresses the setting of multiple \experts{},
working under the assumption that at least one of them is sufficiently
accurate, and seeking to develop online algorithms with a vanishing regret.
To the best of our knowledge, this is the first paper that aims at this
direction.

Most closely related to the current paper are the papers of
\cite{Lykouris2018competitive, Rohatgi2020near, Wei2020better}
on online paging with predictions.
The authors of these papers
stick to the setting of a (single) NAT \expert{} and quantify the \expert{}'s
accuracy by means of the \emph{$L_{1}$-norm}.
Specifically, taking
$\{ \predict{t} \}_{t \in [T]}$
to be the sequence of NAT predictions, they define the \expert{}'s cumulative
prediction error to be
$\LpNorm = \sum_{t} \vert \predict{t} - \NextTime{t}{\sigma_{t}} \vert$.
It is easy to see that for any NAT \expert{}, the cumulative prediction error
as defined in (\ref{equation:nat-cumulative-prediction-error}) is never
larger than its $\LpNorm$, while the former can be
$\Omega (T)$-times
smaller.

\sloppy
In particular, Lykouris and Vassilvitskii \cite{Lykouris2018competitive}
design a randomized online paging algorithm whose competitive ratio is at
most
$O \left(
\min \left\{
1 + \sqrt{\LpNorm / \OPT},
\log k
\right\}
\right)$.
Rohatgi \cite{Rohatgi2020near} presents an improved randomized online
algorithm with competitive ratio up-bounded by
$O \left( 
\min \left\{
1 + \frac{\log k}{k} \frac{\LpNorm}{\OPT},
\log k
\right\}
\right)$
and accompany this with a lower bound of
$\Omega \left(
\min \left\{
1 + \frac{1}{k \log k} \frac{\LpNorm}{\OPT},
\log k
\right\}
\right)$.
Notice that the online algorithms presented in
\cite{Lykouris2018competitive, Rohatgi2020near}
belong to the \emph{marking family} of paging algorithms \cite{FiatKLMSY91} and
it can be shown that the competitive ratio of any such algorithm is bounded
away from $1$ even when provided with a fully accurate \expert{}
(consider for example the paging instance defined by setting
$n = 4$,
$k = 2$,
and
$\sigma_{t} = (t \bmod{4}) + 1$
for every
$t \in [T]$).
\par\fussy

Recently, Wei \cite{Wei2020better} advanced the state of the art of this
problem further, presenting a randomized
$O \left(
\min \left\{
1 + \frac{1}{k} \frac{\LpNorm}{\OPT},
\log k
\right\}
\right)$-competitive
online paging algorithm.
To do so, Wei analyzes an algorithm called \emph{BlindOracle}, that can be
viewed as a variant of our \Sim{} algorithm (see
Section~\ref{section:contribution}), and proves that its competitive ratio is
at most
$\min \left\{
1 + O \left( \frac{\LpNorm}{\OPT} \right),
O \left( 1 + \frac{1}{k} \frac{\LpNorm}{\OPT} \right)
\right\}$.
He then plugs this algorithm into the multiplexer of \cite{Blum2000online}
together with an
$O (\log k)$-competitive
off-the-shelf randomized online paging algorithm to obtain the promised
competitive ratio.
Notice that the bound that Wei establishes on the competitive ratio of
BlindOracle immediately implies an
$O (\LpNorm)$
bound on the regret of this algorithm.
As such, Theorem~\ref{theorem:bound-regret-Sim} can be viewed as a refinement
of Wei's result, bounding the regret as a function of $\eta$ rather than the
weaker measure of $\LpNorm$.

Antoniadis et al.~\cite{Antoniadis2020online} studies online algorithms with
ML predictions in the context of the \emph{metrical task system (MTS)} problem
\cite{BorodinLS1992mts}.
They consider a different type of \expert{} that in each round
$t \in [T]$,
provides a prediction $\hat{s}_{t}$ for the state $s_{t}$ of an optimal
offline algorithm, measuring the prediction error by means of 
$\StateDistance = \sum_{t \in [T]} \operatorname{dist}(s_{t}, \hat{s}_{t})$,
where
$\operatorname{dist}(\cdot, \cdot)$
is the distance function of the underlying metric space.
It is well known that any paging instance $\mathcal{I}$ can be transformed
into an MTS instance $\mathcal{I}_{\text{MTS}}$.
Antoniadis et al.\ prove that the prediction sequence
$\{ \predict{t} \}_{t \in [T]}$
of a NAT \expert{} for $\mathcal{I}$ can also be transformed into a prediction
sequence
$\{ \hat{s}_{t} \}_{t \in [T]}$
for $\mathcal{I}_{\text{MTS}}$.
However, the resulting prediction error $\StateDistance$ of the latter sequence
is incomparable to the prediction error $\LpNorm$ of the former;
this remains true also for the stronger notion of cumulative prediction error
as defined in (\ref{equation:nat-cumulative-prediction-error}).

Online algorithms with access to multiple \experts{} have been studied by Gollapudi and Panigrahi for the \emph{ski rental} problem \cite{GollapudiP19}.
Among other results, they prove that a competitive ratio of
$\alpha = \frac{4}{3}$
(resp.,
$\alpha = \frac{\sqrt{5} + 1}{2}$)
can be achieved by a randomized (resp., deterministic) online algorithm that has access to two \experts{} assuming that at least one of them provides an accurate prediction for the number of skiing days.
Notice that the length of the request sequence in the ski rental problem is inherently bounded by the cost of buying the ski gear;
this is in contrast to the paging problem considered in the current paper, where much of the challenge comes from the unbounded request sequence.

The reader may have noticed that some of the terminology used in the current paper is borrowed from the \emph{online learning} domain \cite{Cesa-BianchiL2006book}.
The main reason for this choice is that the research objectives of the current paper are, to a large extent, more in line with the objectives common to the online learning literature than they are in line with the objectives of the literature on online computation.
In particular, as discussed already, we measure the quality of our online algorithms by means of their regret (rather than competitiveness), indicating that the online algorithm can be viewed as a decision maker that tries to learn the best offline algorithm.

\subsection{Paper's Organization}
\label{section:organization}
The remainder of this paper is organized as follows.
In Section~\ref{section:comparison}, we refine the notion of cumulative prediction error as defined in (\ref{equation:nat-cumulative-prediction-error}) and compare the refined notion with the number of inversions used in some of the related literature
\cite{Rohatgi2020near, Wei2020better}.
The analysis of the \Sim{} algorithm (using a single \expert{}), leading to the proof of Theorem~\ref{theorem:bound-regret-Sim}, is carried out in Section~\ref{section:single-expert-oracle}.
Section~\ref{section:multiple-experts-oracle} is then dedicated to the setting of multiple (NAT) \experts{}.
Most of its content is found in Section~\ref{section:bandit}, devoted to the the bandit access model, where we also establish Theorem~\ref{theorem:bound-regret-bandit-access}.
As discussed in Section~\ref{section:related-work}, Theorem~\ref{theorem:bound-regret-full-information-access}, dealing with the full information access model, follows from Theorem~\ref{theorem:bound-regret-Sim} combined with a technique of \cite{Blum2000online};
this is explained in more detail in Section~\ref{section:full_info_oracle}.


\section{Measurements of Prediction Errors}
\label{section:comparison}

The measurement of the prediction errors plays an important role in the study on online paging algorithms augmented by predictions. This part makes a comparison between different measurements for the scenario where there is a single NAT predictor. To avoid ambiguity, in this part we use $\ErrorRounds$ to represent the measurement defined in Eq.~\eqref{equation:nat-cumulative-prediction-error} for the single predictor $j$. In the following, the superscript $j$ for the predictor is omitted for convenience.

%
%
In the analysis of \cite{Wei2020better}, the prediction errors is measured with the number of \emph{inverted pairs}. For a pair of two rounds $\lbrace t, t' \rbrace$, we say it is an inverted pair if $\NextTime{t}{\sigma_{t}} < \NextTime{t'}{\sigma_{t'}}$ and $\predict{t} \geq \predict{t'}$. Let $\INV$ be the set of all the inverted pairs, and define $\InvertedPairs = \vert \INV \vert$. To compare the measurement $\InvertedPairs$ with $\ErrorRounds$, we also define the following notations.
\begin{align*}
    \InvertedRounds \doteq& \Big\vert \lbrace t \, \vert \, \exists t' \text{ s.t.~} \lbrace t, t' \rbrace \in \INV \rbrace \Big\vert \\
    \RoundsWithErrorAndInversion \doteq& \Big\vert \lbrace t \, \vert \, \NextTime{t}{\sigma_{t}} \neq \predict{t} \, \wedge \, \exists t' \text{ s.t.~} \lbrace t, t' \rbrace \in \INV \rbrace \Big\vert
\end{align*}

First, it trivially holds that $\InvertedRounds \leq 2 \cdot \InvertedPairs$, and $\InvertedPairs$ can be $\Omega(T)$ times larger than $\InvertedRounds$. To see the second claim, consider the following sequence $\sigma$ of requested pages and a prediction sequence $\pi$.
\begin{equation*}
    \sigma_{t} = \begin{cases}
    1 & \text{if } t \leq \frac{T}{2} \\
    2 & \text{if } t > \frac{T}{2}
    \end{cases} \, ,
    \quad \text{and} \quad
    \pi_{t} = \begin{cases}
    2 & \text{if } t \leq \frac{T}{2} \\
    1 & \text{if } t > \frac{T}{2}
    \end{cases} \, .
\end{equation*}
It can be verified that for a sequence of predictions in the form of NATs that are consistent with the settings above, $\InvertedPairs$ is in the order of $T^{2}$ while $\InvertedRounds$ is in the order of $T$.

Second, we claim that $\InvertedRounds$ and $\ErrorRounds$ are incomparable, which means that there exists an example where $\InvertedRounds$ is $\Omega(T)$ times larger than $\ErrorRounds$, and vice versa. Still, we demonstrate these examples with the sequence $\sigma$ of requests and the prediction sequence $\pi$, and the claims above can be verified after converting the predictions in the form of requests to consistent predictions in the form of NATs. The configuration of the first example is given as follows.
\begin{equation*}
    \sigma_{t} = \begin{cases}
    (t \bmod{(k - 1)}) + 1 & \text{if } 1 < t < T \\
    k & \text{otherwise} 
    \end{cases} \, , 
    \quad \text{and} \quad
    \pi = \begin{cases}
    (t \bmod{(k - 1)}) + 1 & \text{if } t > 2 \\
    k & \text{otherwise}
    \end{cases} \, .
\end{equation*}
The second example is configured as follows.
\begin{equation*}
    \sigma_{t} = (t \bmod{(k - 1)}) + 1 \, , 
    \quad \text{and} \quad
    \pi = \begin{cases}
    ((t - 1) \bmod{(k - 1)}) + 1 & \text{if } t > 1 \\
    k & \text{otherwise}
    \end{cases} \, .
\end{equation*}

\sloppy
Third, it is obvious that
\begin{equation*}
    \RoundsWithErrorAndInversion \leq \min \Big\lbrace \ErrorRounds, \InvertedRounds \Big\rbrace \, .
\end{equation*}
Although in \Cref{section:machine-learning} we define $\eta^{j}$ for every {\expert} $j$ in the form of $\ErrorRounds$ for simplicity, our technique indeed works for the better measurement $\RoundsWithErrorAndInversion$. Therefore, in the technical parts of the current paper, including \Cref{section:single-expert-oracle} and \Cref{section:multiple-experts-oracle}, we use the following refined definition of $\eta^{j}$ by abuse of notation:
\begin{equation*}
    \eta^{j} \doteq \Big\vert \lbrace t \, \vert \, \NextTime{t}{\sigma_{t}} \neq \predictJ{t}{j} \, \wedge \, \exists t' \text{ s.t.~} \lbrace t, t' \rbrace \in \INV^{j} \rbrace \Big\vert \, ,
\end{equation*}
where $\INV^{j} = \lbrace (t, t') \vert \NextTime{t}{\sigma_{t}} < \NextTime{t'}{\sigma_{t'}} \, \wedge \, \predictJ{t}{j} \geq \predictJ{t'}{j} \rbrace$.
\par\fussy


\section{Single NAT \Expert{}}
\label{section:single-expert-oracle}

We start with the NAT \expert{} setting with $M = 1$. Notice that when there is only a single NAT {\expert}, there is no difference between the full-information access model and the bandit access model. Throughout this section, we still omit the superscript $j$ for the index of the {\expert}. 

The algorithm {\Sim} that we consider for this setting simulates {\FitF} with maintaining a value $\predictHat{t}{i}$, which we call the \emph{remedy prediction}, for each round $t \in [T]$ and each page $i \in [n]$. 
In particular, for each page $i \in [n]$, {\Sim} sets
\begin{equation}
\begin{aligned}
    \predictHat{1}{i} =& \begin{cases}
    \predict{1} & \text{if } i = \sigma_{1} \\
    Z + 1 & \text{otherwise}
    \end{cases} \, , \quad \text{and}\\
    \forall \, t \in [2, T] \quad \predictHat{t}{i} =& \begin{cases}
    \predict{t} & \text{if } i = \sigma_{t} \\
    Z & \text{if } \predictHat{t - 1}{i} \leq t \, \wedge \, i \neq \sigma_{t} \, \wedge \, \predictHat{t - 1}{i} \leq \predictHat{t - 1}{\sigma_{t}} < Z \\
    \predictHat{t - 1}{i} & \text{otherwise}\\
    \end{cases}\, , 
\end{aligned} \label{formula_virtual_NAT}
\end{equation}
where $Z > T + n$ is a sufficiently large integer. 
For each round $t$ when a cache miss happens, the algorithm evicts the page $\hat{e}_{t} = i$ that maximizes $\predictHat{t}{i}$, and ties are broken in an arbitrary way. The following statements can be directly inferred from Eq.~\eqref{formula_virtual_NAT}.

\begin{lemma} \label{lemma_virtual_NAT_changes}
The following properties are satisfied for every round
$t \in [T]$:
\begin{itemize}

\item
If
$\predict{t} > \NextTime{t}{\sigma_{t}}$,
then for each round
$t' \in [t, \min\lbrace T, \NextTime{t}{\sigma_{t}} - 1 \rbrace]$,
we have
$\predictHat{t'}{\sigma_{t}} > \NextTime{t'}{\sigma_{t}}$.

    
\item
If
$\predict{t} < \NextTime{t}{\sigma_{t}}$,
then for each round
$t' \in [t, \min\lbrace T, \NextTime{t}{\sigma_{t}} - 1 \rbrace]$,
either
$\predictHat{t'}{\sigma_{t}} < \NextTime{t'}{\sigma_{t}}$
or
$\predictHat{t'}{\sigma_{t}} = Z > \NextTime{t'}{\sigma_{t}}$.
Particularly, if
$t' \leq \predict{t} - 1$,
then
$\predictHat{t'}{\sigma_{t}} < \NextTime{t'}{\sigma_{t}}$.

\item
If
$\predict{t} = \NextTime{t}{\sigma_{t}}$,
then for each round
$t' \in [t, \min\lbrace T, \NextTime{t}{\sigma_{t}} - 1 \rbrace]$,
we have
$\predictHat{t'}{\sigma_{t}} = \NextTime{t'}{\sigma_{t}}$.

\end{itemize}
\end{lemma}

Next, we will analyze the cost incurred by {\Sim} and show that it has a vanishing regret.


\subsection{Definitions and Notations for Analysis} \label{section:defs}

For each round $t$, we use $e_{t}$ and $\hat{e}_{t}$ to represent the pages that are evicted by {\FitF} and {\Sim}, respectively. We say $e_{t} = \perp$ (resp.~$\hat{e}_{t} = \perp$) if {\FitF} (resp.~{\Sim}) does not evict any page.

For each page $i \in [n]$ and each round $t \in [T]$, define $\LastRequest{t}{i}$ to be the last round before $t$ when $i$ is requested. Formally,
\begin{equation}
    \LastRequest{t}{i} \doteq \begin{cases}
    \max\lbrace t' < t \mid \sigma_{t'} = i \rbrace & \text{if } \exists\, t' \in [1, t) \text{ s.t.~} \sigma_{t'} = i \\
    \NULL & \text{otherwise}
    \end{cases} \, . \label{formula_last_request}
\end{equation}
The following results can be inferred from Eq.~\eqref{formula_last_request} and \Cref{lemma_virtual_NAT_changes}.


For any round $t \in [T]$, let $C_{t}$ and $\hat{C}_{t}$ be the \emph{cache profiles} incurred by {\FitF} and {\Sim}, respectively. More specifically, $C_{1}$ and $\hat{C}_{1}$ represent the cache items given at the beginning. To provide tools for the more complicated scenario where there are multiple {\experts}, the analysis in this section is carried out without assuming that $C_{1} = \hat{C}_{1}$. For each $t \in [T - 1]$, the cache profile $C_{t}$ (resp.~$\hat{C}_{t}$) is updated to $C_{t + 1}$ (resp.~$\hat{C}_{t + 1}$) immediately after {\FitF} (resp.~{\Sim}) has processed the request $\sigma_{t}$. The cache profiles of {\FitF} and {\Sim} after serving $\sigma_{T}$ are denoted by $C_{T + 1}$ and $\hat{C}_{T + 1}$, respectively.

Denote the intersection between the cache profiles at each round $t \in [T]$ by $I_{t} = C_{t} \cap \hat{C}_{t}$. Define the \emph{distance} between the cache profiles to be $d_{t} = k - \lvert I_{t} \rvert$. We use $\delta_{t}$ to represent the difference in the costs between {\FitF} and {\Sim} for serving $\sigma_{t}$. Formally,
\begin{equation*}
    \delta_{t} \doteq 1_{\sigma_{t} \notin \hat{C}_{t}} - 1_{\sigma_{t} \notin C_{t}} \, .
\end{equation*}

Define $\HD{x}{y}$ for $x \in \mathbb{Z}$, $y \in \mathbb{Z}$ to be the set of rounds $t \in [T]$ where the $d_{t + 1} - d_{t} = x$ and $\delta_{t} = y$. Let $\HD{x}{} = \bigcup_{y} \HD{x}{y}$ and $\HD{}{y} = \bigcup_{x} \HD{x}{y}$.

For a round $t \in [T]$, we say that $t$ is a \emph{troublemaker} if and only if $t$ satisfies
\begin{equation}
\big( \hat{e}_{t} \neq \perp \big) \: \wedge \: \big( \hat{e}_{t} \in I_{t} \big) \: \wedge \: \NextTime{t}{\hat{e}_{t}} \in [T] \: \wedge \: \big( \hat{e}_{t} \in C_{\NextTime{t}{\hat{e}_{t}}} \big) \: \wedge \: \big( e_{t} = \perp \: \vee \: e_{t} \notin I_{t} \big) \, . \label{formula_troublemaker_def}
\end{equation}
The set of troublemaker rounds is denoted by $\Gamma$. For any troublemaker $\gamma \in \Gamma$ and any round $t \in (\gamma, T]$, we say $\gamma$ is \emph{active} at $t$ if $t < A_{\gamma}(\hat{e}_{\gamma})$. The set of troublemakers that are active at $t$ is denoted by $\Gamma_{t} \subseteq \Gamma \cap [t - 1]$. The \emph{active period} $[\gamma + 1, A_{\gamma}(\hat{e}_{\gamma}) - 1]$ of $\gamma$ is denoted by $\ActI_{\gamma}$.

\subparagraph*{Preliminary results.} The following results can be directly inferred from the definitions above.

\begin{lemma} \label{lemma_last_to_next}
For any round $t$ and any page $i$, the following properties are satisfied.
\begin{itemize}
    \item If $\LastRequest{t}{i} = \NULL$ and $i \neq \sigma_{t}$, then $\predictHat{t}{i} = Z + 1$, and vice versa.
    \item The equality $\predictHat{t}{i} = \NextTime{t}{i}$ holds if $\predict{r} = \NextTime{r}{\sigma_{r}}$, where $r = \LastRequest{t}{i}$.
\end{itemize}
\end{lemma}

\begin{proof}
The first statement can be proved inductively with using Eq.~\eqref{formula_virtual_NAT}. The first statement implies that if $\predictHat{t}{i} = \NextTime{t}{i}$, then $\LastRequest{t}{i} \neq \NULL$. Therefore, the second statement can be inferred from \Cref{lemma_virtual_NAT_changes}.
\end{proof}

\begin{lemma} \label{lemma_not_many_null}
$\vert \lbrace t \mid \hat{e}_{t} \neq \perp \, \wedge \, \LastRequest{t}{\hat{e}_{t}} = \NULL \rbrace \vert \leq k$.
\end{lemma}

\begin{proof}
For each round $t$ and each page $i \in \hat{C}_{t}$, the equality $\LastRequest{t}{i} = \NULL$ holds only if $i \in \hat{C}_{1}$. The initial cache profile $\hat{C}_{1}$ contains $k$ different pages, and for each page $i \in \hat{C}_{1}$, if there exist two rounds $t, t'$ with $t < t'$ and $\hat{e}_{t} = \hat{e}_{t'} = i$, then $\LastRequest{t'}{i} \geq t \neq \NULL$. This implies that there are at most $k$ rounds $t$ with $\LastRequest{t}{\hat{e}_{t}} = \NULL$. 
\end{proof}

\begin{lemma}\label{lemma_not_many_null2}
$\vert \lbrace t \mid \LastRequest{t}{\sigma_{t}} = \NULL \, \wedge \, t \in \HD{0}{1} \rbrace \vert \leq k$
\end{lemma}

\begin{proof}
For a round $t \in \HD{0}{1}$, we have $\sigma_{t} \in C_{t}$. Since $R_{t}(\sigma_{t}) = \NULL$, we have $\sigma_{t} \in C_{1}$. Then this lemma can be proved in a similar way with \Cref{lemma_not_many_null}.
\end{proof}

\begin{lemma} \label{lemma_observations_on_troublemakers}
For every round $t$ and any troublemaker $\gamma \in \Gamma_{t}$, we have (1) $\hat{e}_{\gamma} \in C_{t} \setminus \hat{C}_{t}$, and (2) $\hat{e}_{\gamma} \neq \hat{e}_{\gamma'}$ for any troublemaker $\gamma' \in \Gamma_{t}$ with $\gamma \neq \gamma'$.
\end{lemma}

\begin{proof}
The first statement can be directly inferred from the definition of active troublemakers. Now consider the second statement. Without loss of generality, we assume that $\gamma < \gamma'$. Then the first statement shows that $\hat{e}_{\gamma} \notin \hat{C}_{\gamma'}$, which means that {\Sim} cannot evict $\hat{e}_{\gamma}$ at round $\gamma'$.
\end{proof}


\subsection{Reducing Cost Analysis to Troublemaker Counting} \label{section:reduction2-troublemaker}

\begin{lemma} \label{lemma_values_scope}
For each round $t \in [T]$, we have $d_{t + 1} - d_{t} \in \lbrace -1, 0, 1 \rbrace$ and $\delta_{t} \in \lbrace -1, 0, 1 \rbrace$.
\end{lemma}

\begin{proof}
Note that both {\FitF} and {\Sim} are \emph{lazy} algorithms, which means that they only fetch the requested pages into the cache when the cache miss happens. Therefore, it trivially holds that $\delta_{t} \in \lbrace -1, 0, 1 \rbrace$. To analyze $d_{t + 1} - d_{t}$, we consider the following cases.
\begin{itemize}

\item
$e_{t} = \hat{e}_{t} = \perp$:
this implies that
$I_{t + 1} = I_{t}$
and
$d_{t + 1} - d_{t} = 0$.

\item
$e_{t} = \perp$, $\hat{e}_{t} \neq \perp$, and $\hat{e}_{t} \in I_{t}$:
this implies that
$I_{t + 1} = (I_{t} \setminus \lbrace \hat{e}_{t} \rbrace ) \cup \lbrace \sigma_{t} \rbrace $, which gives $d_{t + 1} - d_{t} = 0$.

\item $e_{t} = \perp$, $\hat{e}_{t} \neq \perp$, and $\hat{e}_{t} \notin I_{t}$: this gives $I_{t + 1} = I_{t} \cup \lbrace \sigma_{t} \rbrace $, so $d_{t + 1} - d_{t} = -1$.
    \item $e_{t} \neq \perp$, $\hat{e}_{t} \neq \perp$, $e_{t} \notin I_{t}$, and $\hat{e}_{t} \in I_{t}$: in such a case, $I_{t + 1} = (I_{t} \setminus \lbrace \hat{e}_{t} \rbrace) \cup \lbrace \sigma_{t} \rbrace$, which means $d_{t + 1} - d_{t} = 0$.
    \item $e_{t} \neq \perp$, $\hat{e}_{t} \neq \perp$, $e_{t} \in I_{t}$, $\hat{e}_{t} \in I_{t}$, and $e_{t} = \hat{e}_{t}$: it still holds that $I_{t + 1} = (I_{t} \setminus \lbrace \hat{e}_{t} \rbrace) \cup \lbrace \sigma_{t} \rbrace$.
    \item $e_{t} \neq \perp$, $\hat{e}_{t} \neq \perp$, $e_{t} \in I_{t}$, $\hat{e}_{t} \in I_{t}$, and $e_{t} \neq \hat{e}_{t}$: then we have $I_{t + 1} = (I_{t} \setminus \lbrace e_{t}, \hat{e}_{t} \rbrace) \cup \sigma_{t}$. Therefore, $d_{t + 1} - d_{t} = 1$.
\end{itemize}
Notice that the calculation above on $d_{t + 1} - d_{t}$ follows because it always holds that $e_{t} \in I_{t}$ and $\hat{e}_{t} \in I_{t}$, and if $e_{t} \neq \perp$ or $\hat{e}_{t} \neq \perp$, then $\sigma_{t} \notin I_{t}$. We omit the discussion on the cases that are symmetric with the cases above. 
\end{proof}

\Cref{lemma_values_scope} implies that for the sets $\HD{x}{y}$, we only need to consider the parameters $x, y \in \lbrace -1, 0, 1 \rbrace$. The following result on $\HD{x}{y}$ can be inferred from the proof of \Cref{lemma_values_scope}.

\begin{lemma} \label{lemma_sets_description}
It holds that $\HD{1}{} = \HD{1}{0} = \lbrace t \mid e_{t} \neq \perp \bigwedge \hat{e}_{t} \neq \perp \bigwedge e_{t} \neq \hat{e}_{t} \bigwedge \hat{e}_{t} \in I_{t} \bigwedge e_{t} \in I_{t} \rbrace$.
\end{lemma}

\Cref{lemma_sets_description} implies that $\HD{1}{1} = \emptyset$, therefore, $\cost({\Sim}) - \mathtt{OPT}$ can be bounded by $\lvert \HD{}{1} \rvert - \lvert \HD{}{-1} \rvert = \lvert \HD{0}{1} \rvert + \lvert \HD{-1}{1} \rvert - \lvert \HD{}{-1} \rvert$.

\begin{lemma} \label{lemma_up_down}
$\lvert \HD{-1}{} \rvert \leq \lvert \HD{1}{} \rvert + k$.
\end{lemma}

\begin{proof}
By definition, $\HD{-1}{}$ is the set of rounds $t$ with $d_{t + 1} - d_{t} < 0$, and $\HD{1}{}$ is the set of rounds $t$ with $d_{t + 1} - d_{t} > 0$. Since $d_{1} \leq k$ and $d_{t} \geq 0$ for every $t \in [T + 1]$, this proposition holds.
\end{proof}

\Cref{lemma_up_down} allows us to bound $\lvert \HD{-1}{1} \rvert$ with $\lvert \HD{1}{} \rvert$. The following result follows from the mechanisms of {\FitF} and {\Sim} in choosing the page for eviction when cache miss happens.

\begin{lemma} \label{lemma_observing_crossing_property}
For each round $t \in \HD{1}{}$, we have $\NextTime{t}{\hat{e}_{t}} < \NextTime{t}{e_{t}}$ and $\predictHat{t}{e_{t}} \leq \predictHat{t}{\hat{e}_{t}}$.
\end{lemma}

For each round $t \in \HD{1}{}$, we say that $t$ \emph{blames} another round $t'$ specified as follows. Let $r = \LastRequest{t}{\hat{e}_{t}}$ and $r' = \LastRequest{t}{e_{t}}$, then the round blamed by $t$ is
\begin{equation*}
    t' = \begin{cases}
    r & \text{if } ( r \neq \NULL ) \, \wedge \, ( \predict{r} \neq \NextTime{r}{\sigma_{r}} ) \\
    r' & \text{if } \big( r = \NULL \, \vee \, \predict{r} = \NextTime{r}{\sigma_{r}} \big) \, \wedge \, \big( r' \neq \NULL \, \wedge \, \predict{r'} \neq \NextTime{r'}{\sigma_{r'}} \big) \\
    \NULL & \text{otherwise}
    \end{cases}\, .
\end{equation*}

\begin{lemma} \label{lemma_distance1_blame_null}
For each $t \in \HD{1}{}$, let $t'$ be the round blamed by $t$. If $t' = \NULL$, then $\LastRequest{t}{\hat{e}_{t}} = \NULL$. If $\LastRequest{t}{\hat{e}_{t}} \neq \NULL$, then there exists a round $t''$ such that $\lbrace t', t'' \rbrace \in \INV$.
\end{lemma}

\begin{proof}
For the first claim, suppose on the contrary $r = \LastRequest{t}{\hat{e}_{t}} \neq \NULL$. Now consider two cases, $r' \neq \NULL$ and $r' = \NULL$, where $r' = \LastRequest{t}{e_{t}}$.
\begin{itemize}
    \item $r' \neq \NULL$: Since $t' = \NULL$, in such a case we have $\predict{r} = \NextTime{r}{\sigma_{r}}$ and $\predict{r'} = \NextTime{r'}{\sigma_{r'}}$. By \Cref{lemma_last_to_next}, it holds that $\predictHat{t}{\hat{e}_{t}} = \predict{r} = \NextTime{r}{\sigma_{r}} = \NextTime{t}{\hat{e}_{t}}$. Similarly, we have $\predictHat{t}{e_{t}} = \NextTime{t}{e_{t}}$. This conflicts with \Cref{lemma_observing_crossing_property}.
    \item $r' = \NULL$: By \Cref{lemma_last_to_next}, in such a case we have $\predictHat{t}{e_{t}} = Z + 1$, because obviously $e_{t} \neq \sigma_{t}$. As it still holds that $\predictHat{t}{\hat{e}_{t}} = \NextTime{t}{\hat{e}_{t}} < Z$, we have $\predictHat{t}{\hat{e}_{t}} < \predictHat{t}{e_{t}}$, which conflicts with \Cref{lemma_observing_crossing_property}.
\end{itemize}

For the second claim, since $\LastRequest{t}{\hat{e}_{t}} \neq \NULL$, \Cref{lemma_last_to_next} indicates that $\predictHat{t}{\hat{e}_{t}} \leq Z$. Now consider the following two cases.
\begin{itemize}
    \item $\predictHat{t}{\hat{e}_{t}} < Z$: By \Cref{lemma_observing_crossing_property}, in such a case we also have $\predictHat{t}{e_{t}} < Z$. It can be inferred from Eq.~\eqref{formula_virtual_NAT} that $\predictHat{t}{\hat{e}_{t}} = \predict{r}$ and $\predictHat{t}{e_{t}} = \predict{r'}$. Notice that the second equality holds because $\predictHat{t}{e_{t}} < Z$ means that $r' \neq \NULL$. Then by \Cref{lemma_observing_crossing_property}, we have $\predict{r'} = \predictHat{t}{e_{t}} \leq \predictHat{t}{\hat{e}_{t}} = \predict{r}$ and $\NextTime{r}{\sigma_{r}} = \NextTime{t}{\hat{e}_{t}} < \NextTime{t}{e_{t}} = \NextTime{r'}{\sigma_{r'}}$. This means that the pair $\lbrace r, r' \rbrace$ is an inversion. Since $\LastRequest{t}{\hat{e}_{t}} = \NULL$, we have either $t' = r$ or $t' = r'$. By taking
    \begin{equation*}
        t'' = \begin{cases}
        r' & \text{if } t' = r \\
        r & \text{if } t' = r'
        \end{cases} \, ,
    \end{equation*}
    this claim is proved.
    
    \item $\predictHat{t}{\hat{e}_{t}} = Z$: Let $t_{1}$ be the first round in $(r, t]$ so that $\predictHat{t_{1}}{\hat{e}_{t}} = Z$. Then Eq.~\eqref{formula_virtual_NAT} indicates that $\predictHat{t_{1} - 1}{\sigma_{t_{1}}} < Z$. By \Cref{lemma_last_to_next}, we have $r_{1} = \LastRequest{t_{1}}{\sigma_{t_{1}}} \neq \NULL$. Then, $\NextTime{r_{1}}{\sigma_{r_{1}}} = t_{1} < \NextTime{t}{\hat{e}_{t}} = \NextTime{r}{\sigma_{r}}$. Moreover, $\predictHat{t_{1} - 1}{\sigma_{r}} \leq \predictHat{t_{1} - 1}{\sigma_{t_{1}}} < Z$ means that $\predict{r} = \predictHat{t_{1} - 1}{\sigma_{r}}$ and $\predict{r_{1}} = \predictHat{t_{1} - 1}{\sigma_{t_{1}}}$. Therefore, the pair $\lbrace r_{1}, r\rbrace$ is an inversion. Then this claim is established if $t' = r$. This equation holds because (1) $r = \LastRequest{t}{\hat{e}_{t}} \neq \NULL$, and (2) \Cref{lemma_last_to_next} implies that $\NextTime{r}{\sigma_{r}} \neq \predict{r}$, because otherwise $\predictHat{t}{\hat{e}_{t}} = \NextTime{t}{\hat{e}_{t}} < Z$.
\end{itemize}
This completes the proof.
\end{proof}

\begin{lemma} \label{lemma_not_over_blamed}
For each round $t' \neq \NULL$, it can be blamed by at most two rounds in $\HD{1}{}$.
\end{lemma}

\begin{proof}
Suppose that $t'$ is blamed by $t \in \HD{1}{}$ such that $t' = \LastRequest{t}{e_{t}}$. Then for any round $t'' \in \HD{1}{}$ with $t' < t'' < t$, it cannot blame $t'$ by taking $t' = \LastRequest{t''}{e_{t''}}$. This is because if there exists such a round $t''$, then by definition we have $e_{t} = e_{t''}$. In such a case, there must exist a round $\tilde{t} \in (t'', t)$ with $\sigma_{\tilde{t}} = e_{t}$, otherwise $e_{t} \notin C_{t}$. This conflicts with the definition that $\LastRequest{t}{e_{t}}$ is the last round before $t$ when $e_{t}$ is requested. For any $t'' \in \HD{1}{}$ with $t'' > t$, it cannot blame $t'$ by taking $t' = \LastRequest{t''}{e_{t''}}$, either. Still, if $e_{t} = e_{t''}$, there must exist a round $\tilde{t} \in (t, t'')$ with $\sigma_{\tilde{t}} = e_{t}$. In such a case, $\LastRequest{t''}{e_{t''}} \geq \tilde{t} > t > t'$. The case $t' = \LastRequest{t}{\hat{e}_{t}}$ is symmetric with the case above.
\end{proof}


\begin{lemma} \label{lemma_blaming_2error}
It holds that $\lvert \HD{1}{} \rvert  \leq 2 \cdot \eta + k$ and $\lvert \HD{-1}{1} \rvert \leq 2 \cdot \eta + 2k$.
\end{lemma}

\begin{proof}
The statement $\lvert \HD{1}{} \rvert  \leq 2 \cdot \eta + k$ follows from \Cref{lemma_not_many_null}, \Cref{lemma_distance1_blame_null} and \Cref{lemma_not_over_blamed}. The statement $\lvert \HD{-1}{1} \rvert \leq 2 \cdot \eta + 2k$ then follows from \Cref{lemma_up_down}.
\end{proof}

Next, we analyze $\lvert \HD{0}{1} \rvert$ with the notion of troublemakers defined in \Cref{section:defs}. In particular, for two rounds $t, t'$ with $t' < t$, we say $t'$ is the \emph{parent} of $t$ and $t$ is the \emph{child} of $t'$ if $t \in \HD{0}{1}$ and $t' = \max\lbrace t'' < t \mid \sigma_{t} \in \hat{C}_{t''} \rbrace$.


\begin{lemma} \label{lemma_1parent_1child}
Every round $t \in \HD{0}{1} \setminus \lbrace \tilde{t} \mid \LastRequest{\tilde{t}}{\sigma_{\tilde{t}}} = \NULL \rbrace$ has one parent $t' \in \Gamma \cup \HD{1}{}$, and any round $t'$ has at most one child.
\end{lemma}

\begin{proof}
Since $t \notin \lbrace \tilde{t} \mid \LastRequest{\tilde{t}}{\sigma_{\tilde{t}}} = \NULL \rbrace$, the page $\sigma_{t}$ is requested at round $r = \LastRequest{t}{i} \in [1, t)$, which gives that $\sigma_{t} \in \hat{C}_{r + 1}$. As $t \in \HD{0}{1}$, we know that $\sigma_{t} \notin \hat{C}_{t}$. Therefore, there must exist a round $t'' \in [r + 1, t)$ with $\sigma_{t} = \hat{e}_{t''} \in \hat{C}_{t''}$. Since $\lbrace t'' \mid t'' < t \, \wedge \, \sigma_{t} \in \hat{C}_{t''} \rbrace \neq \emptyset$, the existence of the parent round of $t$ is ensured.

For the parent $t'$ of round $t$, we know that $\hat{e}_{t'} = \sigma_{t}$, because $\sigma_{t} \in \hat{C}_{t'} \setminus \hat{C}_{t' + 1}$. Then we have $\hat{e}_{t'} \neq \perp$ and $\hat{e}_{t'} \in I_{t'}$, where the second equality holds because $\sigma_{t} \in C_{t}$ and $\sigma_{t''} \neq \sigma_{t}$ for every $t'' \in [t', t - 1]$. Then by the definitions of the parent round and $\HD{0}{1}$, we have $A_{t'}(\hat{e}_{t'}) = t \in [T]$ and $\hat{e}_{t'} \in C_{t}$, which means that $\hat{e}_{t'} \in C_{A_{t'}(\hat{e}_{t'})}$. Therefore, $t' \in \Gamma$ if $e_{t'} = \perp$ or $e_{t'} \notin I_{t'}$. If $e_{t'} \neq \perp$ and $e_{t'} \in I_{t}$, then we have $e_{t'} \neq \hat{e}_{t'}$, because otherwise $\hat{e}_{t'} \notin C_{t}$. By \Cref{lemma_sets_description}, in such a case we have $t' \in \HD{1}{}$.

It remains to prove that any round $t'$ has at most one child. Suppose that $t'$ has two children $t_{1}, t_{2}$ with $t_{1} < t_{2}$. In such a case, $\sigma_{t_{2}} = \hat{e}_{t'} = \sigma_{t_{1}} \in \hat{C}_{t_{1} + 1}$, which conflicts with the definition of the parent round.
\end{proof}

\begin{theorem} \label{theorem_reduction}
$\cost(\mathtt{Sim}) - \mathtt{OPT} \leq 4\eta + 4k + \lvert \Gamma \rvert - \lvert \HD{}{-1} \rvert$.
\end{theorem}

\begin{proof}
By using \Cref{lemma_blaming_2error} and \Cref{lemma_1parent_1child}, we have
\begin{align*}
    \cost(\mathtt{Sim}) - \mathtt{OPT} =&\: \lvert \HD{}{1} \rvert - \lvert \HD{}{-1} \rvert \\
    =&\: \lvert \HD{-1}{1} \rvert + \lvert \HD{0}{1} \rvert - \lvert \HD{}{-1} \rvert \\
    \leq&\: (2\eta + 2k) + (\lvert \HD{1}{} \rvert + \lvert \Gamma \rvert + k) - \lvert \HD{}{-1} \rvert \\
    \leq&\: 4\eta + 4k + \lvert \Gamma \rvert - \lvert \HD{}{-1} \rvert \, .
\end{align*}
In particular, the third transition above follows from \Cref{lemma_not_many_null2}.
\end{proof}
The upper bound on $\lvert \Gamma \rvert - \lvert \HD{}{-1} \rvert$ is studied in the next subsection.


\subsection{Labeling for Troublemakers} \label{section:labeling-troublemakers}

From the high level, the analysis in this part on $\lvert \Gamma \rvert - \lvert \HD{}{-1} \rvert$ is done by showing that each troublemaker $\gamma$ either can be mapped a distinct round in $\HD{}{-1}$, or can be mapped to a round $t < \gamma$ that has a prediction error. We specify the mappings with a procedure called \texttt{Labeling}, which is designed to avoid mapping too many troublemakers to a single prediction error. Notice that \texttt{Labeling} is only used in the analysis, while the paging algorithm is unaware of the output generated by \texttt{Labeling}.


\begin{algorithm}[t]
\caption{Procedure \texttt{Labeling}}
\label{algorithm_labeling}
    \SetAlgoLined
    \KwIn{$\lbrace \NextTime{t}{i} \rbrace_{t \in [T], i \in [n]}$, $\lbrace \predict{t} \rbrace_{t \in [T]}$}
	\KwOut{$\lbrace \lambda_{\gamma} \rbrace_{\gamma \in \Gamma}$}
	\For{each $\gamma \in \Gamma$}{
        Pick an arbitrary element $i \in \Phi_{\gamma}$ and set $\hat{\delta}_{\gamma} = i$;\\
        Set $\lambda_{\gamma}(\gamma + 1) = \hat{\delta}_{\gamma}$; \\
        Set $\tau_{\gamma} = \NextTime{\gamma}{\hat{\delta}_{\gamma}}$; \\
        \If{$\tau_{\gamma} > \NextTime{\gamma}{\hat{e}_{\gamma}}$}{
            Set $\lambda_{\gamma}(t) = \hat{\delta}_{\gamma}$ for all the remaining rounds $t$ in $\ActI{\gamma}$;
        }
        \Else{
            Set $t = \gamma + 2$;\\
            \While{$t < \NextTime{\gamma}{\hat{e}_{\gamma}}$}{
                \If{$\lambda_{\gamma}(t - 1) = \sigma_{t - 1}$}{ \textbf{break}; }
                \If{$\lambda_{\gamma}(t - 1) \neq \hat{e}_{t - 1}$}{
                    Set $\lambda_{\gamma}(t) = \lambda_{\gamma}(t - 1)$;
                }
                \Else{
                    Pick an arbitrary element $i$ from $\Psi_{t}$ and set $\lambda_{\gamma}(t) = i$;\\
                    \If{$\NextTime{t - 1}{\lambda_{\gamma}(t)} > \NextTime{t - 1}{\lambda_{\gamma}(t - 1)}$}{
                        Set $\lambda_{\gamma}(t') = i$ for every round $t' \in [t + 1, \min\lbrace \NextTime{\gamma}{\hat{e}_{\gamma}}, \NextTime{t}{i} \rbrace)$; \\
                        Set $t = \min\lbrace \NextTime{\gamma}{\hat{e}_{\gamma}}, \NextTime{t}{i} \rbrace - 1$; \\
                        \textbf{break};
                    }
                }
                Set $t = t + 1$;
            }
            Set $\lambda_{\gamma}(t') = \perp$ for every $t' \in [t, \NextTime{\gamma}{\hat{e}_{\gamma}})$;
        }
	}
\end{algorithm}

Procedure \texttt{Labeling} takes $\big\langle \lbrace \NextTime{t}{i} \rbrace_{t \in [T], i \in [n]}, \lbrace \predict{t} \rbrace_{t \in [T]} \big\rangle$ as the input, which implicitly encodes the operations of {\FitF} and {\Sim}, and for each troublemaker $\gamma \in \Gamma$, \texttt{Labeling} outputs a labeling function $\lambda_{\gamma}: {\ActI}_{\gamma} \mapsto \big([n] \cup \perp \big)$, which maps each round in the active period ${\ActI}_{\gamma}$ of $\gamma$ to either a page $i$ or an empty value. For each $\gamma \in \Gamma$ and each round $t$ in the active period ${\ActI}_{\gamma}$ with $\lambda_{\gamma}(t) \neq \perp$, we say that page $i = \lambda_{\gamma}(t)$ is labelled by $\gamma$. Procedure \texttt{Labeling} is presented in Algorithm \ref{algorithm_labeling} with notions defined as follows.
\begin{align*}
    \forall t \in [T]:& \qquad
    \ActL_{t} \doteq \bigcup_{\gamma \in \Gamma_{t}} \lambda_{\gamma}(t) \, , \quad \text{and} \quad
    \Phi_{t} \doteq \hat{C}_{t} \setminus \Big( C_{t} \cup \ActL_{t} \Big) \, , \\
    \forall t \in [2, T]:& \qquad
    \Psi_{t} \doteq \hat{C}_{t} \setminus \Big( C_{t} \cup \ActL_{t - 1} \Big) \, .
\end{align*}
Briefly speaking, for every $\gamma \in \Gamma$, \texttt{Labeling} picks an arbitrary page $i = \hat{\delta}_{\gamma}$ from $\Phi_{\gamma}$ and labels $i$ with $\gamma$ for the first round in the active period of $\gamma$, which means setting $\lambda_{\gamma}(\gamma + 1) = i$. For convenience, the NAT of $i$ after $\gamma$ is denoted by $\tau_{\gamma}$. Then we consider the following cases.
\begin{itemize}
    \item $\tau_{\gamma} > \NextTime{\gamma}{\hat{e}_{\gamma}}$: Then the label on $\hat{\delta}_{\gamma}$ is kept throughout the active period $\ActI_{\gamma}$ of $\gamma$.
    \item $\tau_{\gamma} < \NextTime{\gamma}{\hat{e}_{\gamma}}$ and $\hat{\delta}_{\gamma} \in \hat{C}_{\tau_{\gamma}}$: This means that $\hat{\delta}_{\gamma}$ is not evicted by {\Sim} before its NAT after $\gamma$. In such a case, the label on $\hat{\delta}_{\gamma}$ is kept until the last round before its NAT.
    \item $\tau_{\gamma} < \NextTime{\gamma}{\hat{e}_{\gamma}}$ and $\hat{\delta}_{\gamma} \notin \hat{C}_{\tau_{\gamma}}$: In such a case, a labelled page is evicted by {\Sim} before its NAT after $\gamma$. For each round $t$ with such an eviction, we label a new page at round $t + 1$ in $\Psi_{t + 1}$ with $\gamma$. We stop labelling new pages either when the labelled page is requested, or the NAT of the previous labelled page after the previous round is less than the NAT of the current labelled page.
\end{itemize}

Before describing how Procedure \texttt{Labeling} is applied to map each troublemaker to a round with a prediction error or a round in $\HD{}{-1}$, we first prove that this procedure is consistent by showing that $\Phi_{t}$ (resp.~$\Psi_{t}$) is not empty whenever we need to find a new page to label from $\Phi_{t}$ (resp.~$\Psi_{t}$). For every troublemaker round $t \in \Gamma$, define
\begin{equation*}
\zeta_{t} = \begin{cases}
\sigma_{t} & \text{if } e_{t} = \perp \\
e_{t} & \text{otherwise}
\end{cases} \, .
\end{equation*}
Then we have the following results.

\begin{lemma} \label{lemma_special_page_property}
For each troublemaker round $t \in \Gamma$, we have (1) $\zeta_{t} \neq \perp$, (2) $\zeta_{t} \in C_{t} \setminus \hat{C}_{t}$, and (3) for any $\gamma \in \Gamma_{t}$, we have $\zeta_{t} \neq \hat{e}_{\gamma}$.
\end{lemma}

\begin{proof}
Claim (1) and (2) are obvious. Now consider claim (3). Since $\gamma < t < \NextTime{\gamma}{\hat{e}_{\gamma}}$, $\sigma_{t} \neq \hat{e}_{\gamma}$, because otherwise $t = \NextTime{\gamma}{\hat{e}_{\gamma}}$. By the definition of troublemakers, $\hat{e}_{\gamma} \in C_{\NextTime{\gamma}{\hat{e}_{\gamma}}}$, so for any $t \in [\gamma + 1, \NextTime{\gamma}{\hat{e}_{\gamma}} - 1]$, it holds that $e_{t} \neq \hat{e}_{\gamma}$.
\end{proof}

\begin{lemma} \label{lemma_availability_on_troublemaker}
For each troublemaker round $t \in \Gamma$, it holds that $|\Phi_{t} | \geq 1$.
\end{lemma}

\begin{proof}
By \Cref{lemma_observations_on_troublemakers} and \Cref{lemma_special_page_property}, we have
\begin{equation*}
C_{t} \setminus \hat{C}_{t} \supseteq \lbrace \hat{e}_{\gamma} \rbrace_{\gamma \in \Gamma_{t}} \cup \lbrace \zeta_{t} \rbrace\, .
\end{equation*}
Still by \Cref{lemma_special_page_property}, it follows that $\zeta_{t} \neq \hat{e}_{\gamma}$ for every $\gamma \in \Gamma_{t}$, then
\begin{equation*}
|C_{t} \setminus \hat{C}_{t}| \geq \big|\lbrace \hat{e}_{\gamma} \rbrace_{\gamma \in \Gamma_{t}} \big| + 1 = \big|\Gamma_{t} \big| + 1 \, .
\end{equation*}
Because $|C_{t}| = |\hat{C}_{t}|$, we have $|\hat{C}_{t} \setminus C_{t}| \geq \Big|\Gamma_{t} \Big| + 1 $. Since for every $\gamma$ with $t \in {\ActI}_{\gamma}$, we have $\gamma \in \Gamma_{t}$, then it always holds that
\begin{equation*}
\Big| \ActL_{t} \Big| \leq |\Gamma_{t}|\,.
\end{equation*}
This finishes the proof.
\end{proof}

\begin{lemma} \label{lemma_labelled_not_optimal}
For each round $t \in [T]$, we have $C_{t} \cap \ActL_{t} = \emptyset$.
\end{lemma}

\begin{proof}
This proposition is true because for any round $t$ when we find a new page $i$ to label, $i \notin C_{t}$, and the label on any page $i$ is cancelled before $i$ is requested.
\end{proof}

\begin{lemma} \label{lemma_availablity_on_eviction}
For every $\gamma \in \Gamma$ and every $t \in [\gamma + 2, A_{\gamma}(\hat{e}_{\gamma}) - 1]$, if $\hat{e}_{t - 1} = \lambda_{\gamma}(t - 1)$, then $\Psi_{t} \neq \emptyset$.
\end{lemma}

\begin{proof}
First, consider the case where $\Phi_{t - 1} \neq \emptyset$. For each page $i \in \Phi_{t - 1}$, we know $i \neq \hat{e}_{t - 1}$ because $i \notin \ActL_{t - 1}$ while $\hat{e}_{t - 1} = \lambda_{\gamma}(t - 1) \in \ActL_{t - 1}$. This gives $i \in \hat{C}_{t}$. Moreover, we have $i \neq \sigma_{t - 1}$ because $\hat{e}_{t - 1} \neq \perp$. This gives $i \notin C_{t}$. Therefore, $i \in \hat{C}_{t} \setminus (C_{t} \cup \ActL_{t - 1}) = \Psi_{t}$.

Now consider $\Phi_{t - 1}$ is empty. In such a case, $|\hat{C}_{t - 1} \setminus C_{t - 1}| = |\ActL_{t - 1}|$. \Cref{lemma_observations_on_troublemakers} indicates that for every $\gamma' \in \Gamma_{t'}$ with $t' \in \ActI_{\gamma'}$, we have $\hat{e}_{\gamma'} \in C_{t'} - \hat{C}_{t'}$. Then $|\hat{C}_{t - 1} \setminus C_{t - 1}| = |\ActL_{t - 1}|$ implies that $C_{t - 1} \setminus \hat{C}_{t - 1} = \lbrace \hat{e}_{\gamma'} \rbrace_{\gamma' \in \Gamma_{t - 1}}$. By the definition of active troublemakers, $\sigma_{t - 1} \notin C_{t - 1} - \hat{C}_{t - 1}$. Also, we have $\sigma_{t - 1} \notin I_{t - 1} = C_{t - 1} \cap \hat{C}_{t - 1}$, because $\hat{e}_{t - 1} \neq \perp$. Therefore, $\sigma_{t - 1} \notin C_{t - 1}$, which means that $e_{t - 1} \neq \perp$. Still by the definition of the troublemakers, we have $e_{t - 1} \notin \lbrace \hat{e}_{\gamma'} \rbrace_{\gamma' \in \Gamma_{t - 1}} = C_{t - 1} \setminus \hat{C}_{t - 1}$. Thus, $e_{t - 1} \in I_{t - 1} \subseteq \hat{C}_{t - 1}$. By \Cref{lemma_labelled_not_optimal}, we have $e_{t - 1} \notin \ActL_{t - 1}$ and $e_{t - 1} \neq \hat{e}_{t - 1}$. Putting $e_{t - 1} \neq \hat{e}_{t - 1}$ and $e_{t - 1} \in \hat{C}_{t - 1}$ together, we get $e_{t - 1} \in \hat{C}_{t} - C_{t}$. Therefore, $e_{t - 1} \in \Psi_{t}$.
\end{proof}

\Cref{lemma_availability_on_troublemaker} and \Cref{lemma_availablity_on_eviction} ensure the consistency of Procedure \texttt{Labeling}. The following lemma gives an important property of this procedure that for every round $t$ and every active troublemaker $\gamma \in \Gamma_{t}$, there exists at most one page labelled by $\gamma$.

\begin{lemma} \label{lemma_unique_labelled_page}
For every round $t$ and any page $i \in [n]$, there exists at most one troublemaker $\gamma \in \Gamma_{t}$ so that $\lambda_{\gamma} = i$.
\end{lemma}

\begin{proof}
We prove this proposition inductively. It trivially holds for $t = 1$ because $\Gamma_{1} = \emptyset$. Suppose that it holds for the round $t \geq 1$. 
For any page $i$, if $i \in \ActL_{t + 1} \setminus \ActL_{t}$, then either $t \in \Gamma$, or there exists some $\gamma \in \Gamma_{t}$ with $\hat{e}_{t} = \lambda_{\gamma}(t)$. For the former case, at most one page is chosen from $\Phi_{t}$ and gets labelled by $t$; for the latter case, by the induction hypothesis, there exists at most one troublemaker $\gamma$ satisfying $\hat{e}_{t} = \lambda_{\gamma}(t)$, which makes one page in $\Psi_{t + 1}$ get labelled. Since $\Phi_{t} \cap \ActL_{t} = \emptyset$ and $\Psi_{t + 1} \cap \ActL_{t} = \emptyset$, we only need to prove that these two cases cannot happen at the same time. If $t \in \Gamma$, then by definition we have $\hat{e}_{t} \in I_{t}$. \Cref{lemma_labelled_not_optimal} shows that $I_{t} \cap \ActL_{t} = \emptyset$, which implies that $\hat{e}_{t} \neq \lambda_{\gamma}(t)$ for any $\gamma \in \Gamma_{t}$. Therefore, this proposition holds.
\end{proof}

The proof of \Cref{lemma_unique_labelled_page} gives the following byproduct.

\begin{lemma} \label{lemma_one_new_labeled}
For each round $t \in [2, T]$, we have $\vert \ActL_{t} \setminus \ActL_{t - 1} \vert \leq 1$.
\end{lemma}

Let $t$ be an arbitrary round with $\hat{e}_{t} \neq \perp$. Suppose that there exists a page $i \in \hat{C}_{t}$ satisfying $\NextTime{t}{i} > \NextTime{t}{\hat{e}_{t}}$, $i \notin \ActL_{t}$, and $i \in \ActL_{t'}$ for every $t' \in [t + 1, \min\lbrace \NextTime{t}{\hat{e}_{t}}, t^{\circ} \rbrace]$, where
\begin{equation}
    t^{\circ} = \begin{cases}
    \min \lbrace t'' \mid t'' > t \: \wedge \: \hat{e}_{t''} = i \rbrace & \text{if } \exists t'' \in (t + 1, T] \text{ s.t.~} \hat{e}_{t''} = i \\
    \infty & \text{otherwise}
    \end{cases} \, . \label{formula_abettor_def}
\end{equation}
\Cref{lemma_one_new_labeled} implies that such a page $i$ is unique if it exists. In such a case, we say that the page $i$ is the \emph{competitor} of $\hat{e}_{t}$, and the round $t$ has an \emph{abettor} round $t^{*}$ specified as follows. Let $r = \LastRequest{t}{\hat{e}_{t}}$ and $r' = \LastRequest{t}{i}$, then the abettor of $t$ is
\begin{equation}
    t^{*} = \begin{cases}
    r & \text{if }  r \neq \NULL  \, \wedge \,  \NextTime{r}{\sigma_{r}} \neq \predict{r}  \\
    r' & \text{if } \big(  r = \NULL  \, \vee \,  \NextTime{r}{\sigma_{r}} = \predict{r}  \big) \, \wedge \, \big( r' \neq \NULL \, \wedge \,  \NextTime{r'}{\sigma_{r'}} \neq \predict{r'}  \big) \\
    \NULL & \text{otherwise}
    \end{cases}\, . \label{formula_definition_abettor}
\end{equation}

\begin{lemma} \label{lemma_crossing_property}
For any round $t^{*} \in [T]$, the number of rounds in $\big\lbrace t \, \vert \, \hat{e}_{t} \neq \perp \, \wedge \, \LastRequest{t}{\hat{e}_{t}} \neq \NULL \, \wedge \, t^{*} \text{ is the abettor of } t \big\rbrace$ is at most two.
\end{lemma}

\begin{proof}
It can be proved in a similar way with \Cref{lemma_not_over_blamed} that $\big\lbrace t \, \vert \, \hat{e}_{t} \neq \perp \, \wedge \, \LastRequest{t}{\hat{e}_{t}} \neq \NULL \, \wedge \, t^{*} \text{ is the abettor of } t \, \wedge \, t^{*} = \LastRequest{t}{\hat{e}_{t}}\big\rbrace$ contains at most one round. It remains to prove that $\big\lbrace t \, \vert \, \hat{e}_{t} \neq \perp \, \wedge \, \LastRequest{t}{\hat{e}_{t}} \neq \NULL \, \wedge \, t^{*} \text{ is the abettor of } t \, \wedge \, t^{*} = \LastRequest{t}{i}\big\rbrace$ contains at most a single round, where $i$ is the competitor of $\hat{e}_{t}$ as defined in Eq.~\eqref{formula_abettor_def}. Notice that different from the case considered in the proof of \Cref{lemma_not_over_blamed}, the page $i$ may not be evicted by {$\FitF$}. Therefore, we need to utilize the properties of procedure \texttt{Labeling} to prove the uniqueness of the round $t$ which satisfies $t^{*} = \LastRequest{t}{i}$.

\begin{claim}\label{claim_competitor_property}
If a round $t$ has an abettor $t^{*} = \LastRequest{t}{i}$ with $i$ being the competitor of $\hat{e}_{t}$ and $r = \LastRequest{t}{\hat{e}_{t}} \neq \NULL$, then it holds that $\predictHat{t}{i} < Z$ and $\predictHat{t'}{i} = Z$ for any round $t' \in [\NextTime{t}{\hat{e}_{t}}, \NextTime{t}{i})$.
\end{claim}

\begin{subproof}
\qedlabel{claim_competitor_property}
Since $t^{*} = \LastRequest{t}{i}$ and $r \neq \NULL$, we have $\NextTime{r}{\sigma_{r}} = \predict{r}$, which by \Cref{lemma_last_to_next} gives $\NextTime{t}{\hat{e}_{t}} = \predictHat{t}{\hat{e}_{t}} < Z$. Following \Cref{lemma_observing_crossing_property}, we have $\predictHat{t}{i} \leq \NextTime{t}{\hat{e}_{t}}$ and $\predictHat{t}{i} < Z$. If there exists a round $t'' \in (t, \NextTime{t}{\hat{e}_{t}})$ so that $\predictHat{t''}{i} = Z$, then it follows from Eq.~\eqref{formula_virtual_NAT} that for any round $t' \in [\NextTime{t}{\hat{e}_{t}}, \NextTime{t}{i})$, it holds that $\predictHat{t'}{i} = Z$, which means that this proposition holds. 
We proceed to prove that $\predictHat{\NextTime{t}{\hat{e}_{t}}}{i} = Z$ if $\predictHat{t''}{i} \neq Z$ holds for every round $t'' \in (t, \NextTime{t}{\hat{e}_{t}})$. In such a case, it can be inferred from Eq.~\eqref{formula_virtual_NAT} that $\predictHat{t''}{i} = \predictHat{t}{i}$. Since $\NextTime{r}{\sigma_{r}} = \predict{r}$, \Cref{lemma_last_to_next} indicates that $\predictHat{t''}{\hat{e}_{t}} = \NextTime{t''}{\hat{e}_{t}} = \NextTime{t}{\hat{e}_{t}} = \predictHat{t}{\hat{e}_{t}}$ and $\predictHat{t''}{\hat{e}_{t}} < Z$ hold for any round $t'' \in (t, \NextTime{t}{\hat{e}_{t}})$. Combining $\predictHat{t''}{\hat{e}_{t}} = \predictHat{t}{\hat{e}_{t}}$ with $\predictHat{t''}{i} = \predictHat{t}{i}$ gives $\predictHat{t''}{\hat{e}_{t}} \geq \predictHat{t''}{i}$. Therefore, $\predictHat{\NextTime{t}{\hat{e}_{t}} - 1}{i} \leq \predictHat{\NextTime{t}{\hat{e}_{t}} - 1}{\hat{e}_{t}} < Z$. Moreover, we have $\predictHat{\NextTime{t}{\hat{e}_{t}} - 1}{i} < \NextTime{t}{\hat{e}_{t}}$ because $\NextTime{t}{\hat{e}_{t}} \geq \predictHat{t}{i} = \predictHat{t''}{i}$ holds for every $t'' \in (t, \NextTime{t}{\hat{e}_{t}})$. By Eq.~\eqref{formula_virtual_NAT}, the conditions for $\predictHat{\NextTime{t}{\hat{e}_{t}}}{i} = Z$ are all satisfied. Thus, the equality $\predictHat{t'}{i} = Z$ holds for any round $t' \in [\NextTime{t}{\hat{e}_{t}}, \NextTime{t}{i})$.
\end{subproof}

Then we have the following arguments. 
\begin{itemize}
    \item For any round $t' \in (t^{*}, t)$, the round $t^{*}$ cannot be the abettor of $t'$ with taking $t^{*} = \LastRequest{t'}{i}$, because otherwise,
    \begin{itemize}
        \item if $\NextTime{t'}{\hat{e}_{t'}} \geq t$, then by the definition of abettors, we have $i \in \ActL_{t}$, which conflicts with the requirement in the definition of abettors; else
        \item if $\NextTime{t'}{\hat{e}_{t'}} < t$, then we get $\predictHat{t}{i} = Z$ with using the second statement in \Cref{claim_competitor_property}, which conflicts with the first statement in \Cref{claim_competitor_property}.
    \end{itemize}
    
    \item For any round $t' \in (t, \min\lbrace \NextTime{t}{\hat{e}_{t}}, t^{\circ} \rbrace]$, the round $t^{*}$ cannot be the abettor of $t'$ with taking $t^{*} = \LastRequest{t'}{i}$ because by definition, $i \in \ActL_{t'}$.
    
    \item For any round $t' \in [\NextTime{t}{\hat{e}_{t}} + 1, t^{\circ})$ with $\hat{e}_{t'} \neq \perp$ and $\LastRequest{t'}{\hat{e}_{t'}} \neq \NULL$, the round $t^{*}$ cannot be the abettor of $t'$ with taking $t^{*} = \LastRequest{t'}{i}$. Suppose on the contrary $t^{*}$ is the abettor of $t'$ with $t^{*} = \LastRequest{t'}{i}$. In such a case, the first statement of \Cref{claim_competitor_property} indicates that $\predictHat{t'}{i} < Z$, which conflicts with the second statement of \Cref{claim_competitor_property}.
    
    \item For any round $t' \in [t^{\circ}, T]$, the round $t^{*}$ cannot be the abettor of $t'$ with taking $t^{*} = \LastRequest{t'}{i}$, because $\LastRequest{t'}{i} > t^{\circ} > t^{*}$.
\end{itemize}
This finishes the proof.
\end{proof}

The following result can be proved by following the same line of arguments with the proof of \Cref{lemma_distance1_blame_null}.

\begin{lemma}\label{lemma_abettor_to_error}
Let $t$ be an arbitrary round that has an abettor $t^{*}$. If $t^{*} = \NULL$, then $\LastRequest{t}{\hat{e}_{t}} = \NULL$. If $\LastRequest{t}{\hat{e}_{t}} \neq \NULL$, then there exists a round $t'$ such that $\lbrace t^{*}, t' \rbrace \in \INV$.
\end{lemma}

\begin{remark}
Notice that the statement of \Cref{lemma_abettor_to_error} is consistent because for any round $t$ having an abettor, by definition we have $\hat{e}_{t} \neq \perp$.
\end{remark}

For an arbitrary round $t$, if there exists a page $i$ in $\hat{C}_{t}$ satisfies (1) $\NextTime{t}{i} \leq T$, (2) $i \notin \ActL_{t}$ and (3) $i \in \ActL_{t'}$ for every $t' \in [t + 1, \NextTime{t}{i}]$, we say that $t^{\triangle} = \NextTime{t}{i}$ is the \emph{savior} of $t$.

\begin{lemma} \label{lemma_cancel_out_cost}
If a round $t$ has a savior $t^{\triangle}$, then (1) $t^{\triangle} \in \HD{}{-1}$, and (2) for any $t^{\triangle} \in \HD{}{-1}$, it is the savior of at most one step $t$.
\end{lemma}

\begin{proof}
The first claim trivially follows from the definition of $\HD{}{-1}$. For any round $t' \in [t + 1, t^{\triangle}]$, $t^{\triangle}$ is not the savior of $t'$, because $\sigma_{t^{\triangle}} \in \ActL_{t'}$. Therefore, the second claim holds.
\end{proof}

\begin{theorem} \label{theorem_troublemakers_bounding}
$\lvert \Gamma \rvert - \lvert \HD{}{-1} \rvert \leq 2 \cdot \eta + k$.
\end{theorem}

\begin{proof}
The main idea of this proof is to show that each troublemaker can be mapped to a distinct \emph{broker} round, and each broker either has an abettor or has a savior. 
In particular, we classify the troublemakers $\gamma \in \Gamma$ into the following three categories.
\begin{enumerate}
    \item $\lbrace \gamma \in \Gamma \mid \tau_{\gamma} > A_{\gamma}(\hat{e}_{\gamma}) \rbrace$: Here, the troublemaker $\gamma$ as a round has an abettor $t^{*}$, because $\hat{\delta}_{\gamma} \in \ActL_{t}$ for every step $t \in [\gamma + 1, \min\lbrace \NextTime{\gamma}{\hat{e}_{\gamma}}, t^{\circ} \rbrace]$ where $t^{\circ}$ is defined in the same way with Eq.~\eqref{formula_abettor_def}. In such a case, we say $\gamma$ is the broker of itself.
    \item $\lbrace \gamma \in \Gamma \mid \tau_{\gamma} < A_{\gamma}(\hat{e}_{\gamma}) \wedge \hat{\delta}_{\gamma} \in \hat{C}_{\tau_{\gamma}} \rbrace$: Now the troublemaker $\gamma$ as a round has a savior $\tau_{\gamma}$, because by the definition of the troublemaker, $\NextTime{\gamma}{\hat{e}_{\gamma}} \leq T$, which gives $\tau_{\gamma} \leq T$. Moreover, it holds for every round $t \in [\gamma + 1, A_{\gamma}(\hat{\delta}_{\gamma})]$ that $\hat{\delta}_{\gamma} \in \ActL_{t}$. The broker for such a troublemaker $\gamma$ is also itself.
    \item $\lbrace \gamma \in \Gamma \mid \tau_{\gamma} < A_{\gamma}(\hat{e}_{\gamma}) \wedge \hat{\delta}_{\gamma} \notin \hat{C}_{\tau_{\gamma}} \rbrace$: In such a case, let $i$ be the last page labelled by $\gamma$ and $t$ be the first round with $\lambda_{\gamma}(t) = i$. Since $\hat{\delta}_{\gamma} \notin \hat{C}_{\tau_{\gamma}}$, we have $t - 1 \in \ActI_{\gamma}$. Let $i' = \lambda_{\gamma}(t - 1)$, then we have $i' \in \hat{C}_{t - 1}$ because $\hat{e}_{t - 1} = i'$. Also, we have $i \in \hat{C}_{t - 1}$, because $i$ is chosen from $\Psi_{t} \subseteq \hat{C}_{t} \setminus C_{t}$ and $i \neq \hat{e}_{t - 1}$. Now consider two subcases.
    \begin{itemize}
        \item $\NextTime{t - 1}{i'} < \NextTime{t - 1}{i}$: In such a case, the round $t - 1$ has an abettor $t^{*}$, because we have $i \in \ActL_{t'}$ for every $t' \in [t, \min\lbrace \NextTime{t - 1}{i}, \tau_{\gamma}, t^{\circ} \rbrace]$, which satisfies the requirement in the definition of abettors because $\tau_{\gamma} > \NextTime{t - 1}{i'}$.
        \item $\NextTime{t - 1}{i'} > \NextTime{t - 1}{i}$: It can be proved inductively that $\NextTime{t - 1}{i} < \tau_{\gamma}$, which implies that $\NextTime{t - 1}{i} \leq T$. By definition, it follows that $t - 1$ has a savior $t^{\triangle} = \NextTime{t - 1}{i}$.
    \end{itemize}
    The round $t - 1$ is said to be the \emph{broker} of the troublemaker $\gamma$. \Cref{lemma_unique_labelled_page} ensures that the round $t - 1$ cannot be the broker of two different troublemakers, because $\hat{e}_{t - 1} = \lambda_{\gamma}(t - 1)$. Moreover, by \Cref{lemma_one_new_labeled}, the broker $t - 1$ is not a troublemaker.
\end{enumerate}
To sum up, each troublemaker $\gamma$ can be mapped to a distinct broker $t$, and each broker $t$ either has an abettor or has a savior. 
Then by \Cref{lemma_crossing_property}, \Cref{lemma_abettor_to_error} and \Cref{lemma_cancel_out_cost}, this theorem holds.
\end{proof}

The following result is obtained by combining \Cref{theorem_reduction} and \Cref{theorem_troublemakers_bounding}.

\begin{theorem} \label{theorem_single_expert_bound}
$\cost(\mathtt{Sim}) - \mathtt{OPT} \leq 6\eta + 5k$.
\end{theorem}

Because it is assumed that the number $\eta$ of prediction errors of the single {\expert} satisfies $\eta \in o(T)$, we have $\cost(\mathtt{Sim}) - \mathtt{OPT} \in o(T)$. Therefore, {\Sim} has the vanishing regret when there is a single {\expert}.


\section{Multiple NAT \Experts{}}
\label{section:multiple-experts-oracle}

This section extends the result obtained in \Cref{section:single-expert-oracle} to the more general scenario where there are $M > 1$ NAT {\experts} for both the bandit access model and the full-information model, respectively.

\subsection{Bandit Access Model}
\label{section:bandit}

For the bandit access model, in this part, we design an algorithm called \emph{Sightless Chasing and Switching} ({\SightlessChasing}) and prove that it has the vanishing regret.

\begin{algorithm}[!ht]
\caption{Algorithm {\SightlessChasing}}
\label{algorithm_SCaS}
    \SetAlgoLined
    \KwIn{$\lbrace \sigma_{t} \rbrace_{t \in [T]}$, MBP algorithm $\mathtt{INF}$, initial cache profile $\hat{C}_{1}$}
    \KwOut{$\lbrace \hat{e}_{t} \rbrace_{t \in [T]}$}
    Initialize the MBP algorithm $\mathtt{INF}$ with the number of rounds $\Upsilon = \lceil T / \tau \rceil$ and the set of arms $X = [M]$; \\
    \For{each round $t \in [T]$}{
        \If{$t \bmod \tau = 1$}{
            Invoke $\mathtt{INF}$ to choose a {\expert} $j_{t} \in [M]$;\\
            Set $j = j_{t}$;\\
        }
        \Else{
            Set $j = j_{t'}$ with $t' = \lfloor t / \tau \rfloor \cdot \tau + 1$;\\
        }
        Query the {\expert} $j$ to obtain the prediction $\predictJ{t}{j}$;\\
        \For{each page $i \in [n]$}{
            \If{$i = \sigma_{t}$}{
                Set $\predictHat{t}{i} = \predictJ{t}{j}$; \\
            }
            \ElseIf{$t \bmod \tau = 1$}{
                Set $\predictHat{t}{i} = Z + 1$; \\
            }
            \ElseIf{$\predictHat{t - 1}{i} = t \, \wedge \, i \neq \sigma_{t} \, \wedge \, \predictHat{t - 1}{i} \leq \predictHat{t - 1}{\sigma_{t}} < Z$}{
                Set $\predictHat{t}{i} = Z$;\\
            }
            \Else{
                Set $\predictHat{t}{i} = \predictHat{t - 1}{i}$;\\
            }
        }
        \If{$\sigma_{t} \notin \hat{C}_{t}$}{
            Set $\hat{e}_{t}$ be the page $i \in C_{t}$ that maximizes $\predictHat{t}{i}$ with breaking ties arbitrarily;\\
            Update $\hat{C}_{t + 1} = (\hat{C}_{t} \setminus \lbrace \hat{e}_{t} \rbrace) \cup \lbrace \sigma_{t} \rbrace$;\\
        }
        \Else{
            Set $\hat{e}_{t} = \perp$, and set $\hat{C}_{t + 1} = \hat{C}_{t}$;\\
        }
        \If{$t \bmod \tau = 0$}{
            Set $f = 0$; \\
            \For{each round $t' \in [t - \tau + 1, t]$}{
                \If{$(t' = t - \tau + 1) \, \vee \,
                (\hat{e}_{t'} \neq \perp) \, \vee \,
                \Big( (\hat{e}_{t'} = \perp) \, \wedge \, (t' > t - \tau + 1) \, \wedge \, (\predictHat{t' - 1}{\sigma_{t'}} = Z + 1) \Big) $}{
                    Set $f = f + 1$; \label{algoline_virtual_CAPEX_increment}\\
                }
            }
            Send $\frac{f}{\tau}$ to $\mathtt{INF}$ as the cost incurred by $j_{t - \tau + 1}$ in the epoch $\frac{t}{\tau}$;\\
        }
    }
\end{algorithm}

The procedure of {\SightlessChasing} is described in Algorithm \ref{algorithm_SCaS}. It is assumed that {\SightlessChasing} is provided with blackbox accesses to the online algorithm \emph{Implicitly Normalized Forecaster} (\texttt{INF}) \cite{Audibert2009MPA} for the \emph{Multiarmed Bandit Problem (MBP)}~\cite{Auer2002NMB}. The MBP problem is an online problem defined over $\Upsilon \in \mathbb{Z}_{> 0}$ rounds and a set $X$ of arms. An oblivious adversary specifies a cost function $F_{\upsilon}: X \mapsto [0, 1]$ for each round $\upsilon \in [\Upsilon]$ that maps each arm $x \in X$ to a cost in $[0, 1]$. An algorithm for MBP needs to choose an arm $x_{\upsilon}$ at the beginning of each round $\upsilon \in [\Upsilon]$, and then the cost $F_{\upsilon}(x_{\upsilon})$ incurred by the chosen arm $x_{\upsilon}$ is revealed to the algorithm. The objective of MBP is to minimize the cumulative cost incurred by the chosen arms $\lbrace x_{\upsilon} \rbrace_{\upsilon \in \Upsilon}$.

\begin{theorem}[\cite{Audibert2009MPA}]\label{theorem_INF_regret}
The algorithm \texttt{INF} ensures that the chosen arms $\lbrace x_{\upsilon} \rbrace_{\upsilon \in [\Upsilon]}$ satisfy
\begin{equation*}
    \sum_{\upsilon \in [\Upsilon]} F_{\upsilon}(x_{\upsilon}) - \min_{x^{*} \in X} \sum_{\upsilon \in [\Upsilon]} F_{\upsilon}(x^{*}) \in O\big(\sqrt{|X| \cdot \Upsilon}\big) \, .
\end{equation*}
\end{theorem}

Our algorithm {\SightlessChasing} partitions the rounds into consecutive epochs of length $\tau \in \mathbb{Z}_{> 0}$ and initializes \texttt{INF} by setting $\Upsilon = \big\lceil \frac{T}{\tau} \big\rceil$ and $X = [M]$, which means that each epoch in the online paging problem is mapped to a round in MBP, and each {\expert} is taken as an arm. The choice of the value for $\tau$ is discussed later. At the beginning of the first round $t_{1}^{\upsilon} = (\upsilon - 1)\tau + 1$ in each epoch $\upsilon \in [\Upsilon]$, {\SightlessChasing} accesses \texttt{INF} to pick one {\expert} $j_{t_{1}^{\upsilon}}$. Then, {\SightlessChasing} simulates the algorithm {\Sim}, which is proposed in \Cref{section:single-expert-oracle}, throughout the epoch $\upsilon$ with taking $t_{1}^{\upsilon}$ as its initial round, $\hat{C}_{t_{1}^{\upsilon}}$ as its initial cache profile, and $j_{t_{1}^{\upsilon}}$ as the single {\expert}. At the end of the last round $t_{\tau}^{\upsilon} = \upsilon \cdot \tau$ in epoch $\upsilon$, {\SightlessChasing} sends
\begin{equation}\label{formula_virtual_bandit_cost}
    F_{\upsilon}(j_{t_{1}^{\upsilon}}) = \frac{1}{n}\Big\vert \Big\lbrace t' \in [t_{1}^{\upsilon}, t_{\tau}^{\upsilon}] \Big\vert
    (t' = t_{1}^{\upsilon}) \, \vee \,
    (\hat{e}_{t'} \neq \perp) \, \vee \,
    \Big( (\hat{e}_{t'} = \perp) \, \wedge \, (t' > t_{1}^{\upsilon}) \, \wedge \, (\predictHat{t' - 1}{\sigma_{t'}} = Z + 1) \Big)
    \Big\vert 
\end{equation}
to \texttt{INF} as the cost $F_{\upsilon}(j_{t_{1}})$ of choosing $j_{t_{1}}$ for $\upsilon$. 

Notice that in MBP, the cost functions are generated by an oblivious adversary. We take this setting as a requirement that the cost function $F_{\upsilon}$ for each round $\upsilon$ in MBP should not depend on the arms chosen in the previous rounds $x_{1}, x_{2}, \dots, x_{\upsilon - 1}$. The following result shows that by feeding \texttt{INF} a cost that can be larger than the normalized cost that is actually incurred by {\SightlessChasing} in the epoch, this requirement is satisfied.

\begin{lemma}\label{lemma_stateless_costs}
Let $F_{\upsilon}\Big(j_{t_{1}^{\upsilon}} \Big\vert j_{t_{1}^{1}}, \dots, j_{t_{1}^{\upsilon - 1}}  \Big)$ be the cost sent by {\SightlessChasing} to \texttt{INF} at the end of an arbitrary epoch $\upsilon$ conditioned on the the {\experts} chosen for the previous epochs $j_{t_{1}^{1}}, \dots, j_{t_{1}^{\upsilon - 1}}$. Then for any different sequence of {\experts} $\tilde{j}_{t_{1}^{1}}, \dots, \tilde{j}_{t_{1}^{\upsilon - 1}}$, we have
\begin{equation*}
    F_{\upsilon}\Big(j_{t_{1}^{\upsilon}} \Big\vert j_{t_{1}^{1}}, \dots, j_{t_{1}^{\upsilon - 1}}  \Big) = F_{\upsilon}\Big(j_{t_{1}^{\upsilon}} \Big\vert j_{t_{1}^{1}}', \dots, j_{t_{1}^{\upsilon - 1}}'  \Big) \, .
\end{equation*}
\end{lemma}

\begin{proof}
For the epoch $\upsilon$, a round $t$ is said to be \emph{fresh} if for any earlier round $t' < t$ in $\upsilon$, it holds that $\sigma_{t} \neq \sigma_{t'}$. We first consider the case where are at least $k$ fresh rounds in the epoch, which means that at least $k$ different pages are requested. Let $t$ be the first round after the first $k$ fresh rounds. 
We first consider the interval $[t_{1}^{\upsilon}, t - 1]$. For each page $i \in [n]$ and each round $t' \in [t_{1}^{\upsilon}, t - 1]$, we say $i$ is \emph{marked} at $t'$ if there exists a round $t'' \in [t_{1}^{\upsilon}, t']$ with $\sigma_{t''} = i$, otherwise $i$ is said to be \emph{unmarked} at $t'$. Then it can be inferred from \Cref{lemma_last_to_next} that $\predictHat{t'}{i} = Z + 1$ if $i$ is unmarked at $t'$, and $\predictHat{t'}{i} \leq Z$ if $i$ is marked at $t'$.
Thus, there is no marked page getting evicted before $t$, and any round $t' \in [t_{1}^{\upsilon}, t - 1]$ that satisfies $\hat{e}_{t'} \neq \perp$ must be a fresh round.

\begin{claim}\label{claim_fresh_rounds}
For each round $t' \in [t_{1}^{\upsilon}, t - 1]$, it is counted by Eq.~\eqref{formula_virtual_bandit_cost} if and only if $t'$ is fresh.
\end{claim}

\begin{subproof}
\qedlabel{claim_fresh_rounds}
By definition, the round $t_{1}^{\upsilon}$ is a fresh round. Now consider a round $t' \in [t_{1}^{\upsilon} + 1, t - 1]$. As mentioned above, if $\hat{e}_{t'} \neq \perp$, then $t'$ is fresh. Also, if $t'$ is fresh, then it can be inferred from \Cref{lemma_last_to_next} that $\predictHat{t' - 1}{\sigma_{t'}} = Z + 1$. If $t'$ is not fresh, which means that $\sigma_{t'}$ is marked, then it holds that $\hat{e}_{t'} = \perp$ and $\predictHat{t' - 1}{\sigma_{t'}} \leq Z$. Therefore, this claim holds.
\end{subproof}

Therefore, the contribution of the rounds in $[t_{1}^{\upsilon}, t - 1]$ to $F_{\upsilon}(j_{t_{1}^{\upsilon}})$ is always $\frac{k}{n}$, which is independent of $j_{t_{1}}^{1}, \dots, j_{t_{1}}^{\upsilon - 1}$. A similar result can also be obtained when there are less than $k$ fresh rounds in the epoch $\upsilon$.

At the beginning of round $t$, $\hat{C}_{t}$ contains exactly the first $k$ different pages required in the epoch $\upsilon$, and for each page $i \in \hat{C}_{t}$, the remedy prediction $\predictHat{t}{i}$ is computed only based on $\lbrace \predictJ{t'}{j_{t_{1}^{\upsilon}}} \rbrace_{t' \in [t_{1}^{\upsilon}, t]}$ and $\lbrace \sigma_{t'} \rbrace_{t' \in [t_{1}^{\upsilon}, t]}$. Therefore, for any $t' \geq t$, the decision on $\hat{e}_{t'}$ is independent of the choices over $j_{t_{1}}^{1}, \dots, j_{t_{1}}^{\upsilon - 1}$.

Furthermore, since at round $t - 1$, every page $i \in \hat{C}_{t - 1}$ is marked, then for any round $t' \geq t$ with $\hat{e}_{t} = \perp$, the page $\sigma_{t'}$ has been requested at least once in the interval $[t_{1}^{\upsilon}, t' - 1]$, which means that $\predictHat{t' - 1}{\sigma_{t'}} \leq Z$. This observation is formally stated in the following claim.

\begin{claim}\label{claim_rounds_all_marked}
For any round $t' \geq t$, it is counted by Eq.~\eqref{formula_virtual_bandit_cost} if and only if $\hat{e}_{t'} \neq \perp$.
\end{claim}

Thus, the contribution of the rounds in $[t, \upsilon \cdot \tau]$ to $F_{\upsilon}(j_{t_{1}^{\upsilon}})$ does not depend on $\lbrace j_{t_{1}}^{1}, \dots, j_{t_{1}}^{\upsilon - 1} \rbrace$, either.
\end{proof}

For an epoch $\upsilon \in \Upsilon$ and a {\expert} $j \in [M]$ chosen for the epoch $\upsilon \in \Upsilon$, let $\Evictions_{\upsilon}(j) = \vert \lbrace t \in [t_{1}^{\upsilon}, t_{\tau}^{\upsilon}] \mid \hat{e}_{t} \neq \perp \rbrace \vert$. The following result is also obtained by combining \Cref{claim_fresh_rounds} with \Cref{claim_rounds_all_marked}.

\begin{lemma}\label{lemma_cost_function_proxy}
For each epoch $\upsilon$,
\begin{equation*}
    \Evictions_{\upsilon}(j) \; \leq \; \tau \cdot F_{\upsilon}(j) \; \leq \; \Evictions_{\upsilon}(j) + k \, . 
\end{equation*}
\end{lemma}

\begin{theorem}\label{theorem_regret_multi_oracles}
By taking $\tau = \Big\lfloor T^{\frac{1}{3}} \Big\rfloor$, the regret of {\SightlessChasing} is bounded by $O\Big( kT^{\frac{2}{3}}\sqrt{M} + \eta^{\min}\Big)$.
\end{theorem}

\begin{proof}
\Cref{lemma_stateless_costs} allows us to apply \Cref{theorem_INF_regret} to get the following result.
\begin{equation*}
    \sum_{\upsilon \in [\Upsilon]}F_{\upsilon}(j_{t_{1}^{\upsilon}}) - \min_{j^{*} \in [M]}\sum_{\upsilon \in [\Upsilon]}F_{\upsilon}(j^{*}) \in O\Bigg(\sqrt{M\cdot \frac{T}{\tau}}  \Bigg) \, .
\end{equation*}
\Cref{lemma_cost_function_proxy} indicates that
\begin{equation*}
    \tau \cdot \min_{j^{*} \in [M]}\sum_{\upsilon \in [\Upsilon]}F_{\upsilon}(j^{*}) - \min_{j \in [M]} \sum_{\upsilon \in [\Upsilon]} \Evictions_{\upsilon}(j) \leq k \cdot \Upsilon \, .
\end{equation*}

For any $j \in [M]$, since $\Evictions_{\upsilon}(j)$ is the number of cache misses in the epoch $\upsilon$ by simulating the algorithm {\Sim} with taking $t_{1}^{\upsilon}$ as the initial round, $\hat{C}_{t_{1}^{\upsilon}}$ as the initial cache profile, and $j$ as the single {\expert}, then it can be inferred from \Cref{theorem_single_expert_bound} that
\begin{equation*}
    \Evictions_{\upsilon}(j) - \big\vert \lbrace t \in [t_{1}^{\upsilon}, t_{\tau}^{\upsilon}] \, \vert \, e_{t} \neq \perp \rbrace \big\vert \leq 6 \cdot \eta_{\upsilon}^{j} + 5k \, ,
\end{equation*}
where
\begin{equation*}
    \eta_{\upsilon}^{j} = \Big\vert \big\lbrace t \in [t_{1}^{\upsilon}, t_{\tau}^{\upsilon}] \, \big\vert \, \predictJ{t}{j} \neq \NextTime{t}{\sigma_{t}} \, \wedge \, \exists t' \in [t_{1}^{\upsilon}, t_{\tau}^{\upsilon}] \text{ s.t.~} \lbrace t, t' \rbrace \in \INV^{j} \big\rbrace \Big\vert \, .
\end{equation*}
Since
\begin{equation*}
    \OPT = \sum_{\upsilon \in [\Upsilon]}\big\vert \lbrace t \in [t_{1}^{\upsilon}, t_{\tau}^{\upsilon}] \, \vert \, e_{t} \neq \perp \rbrace \big\vert \, ,
    \quad \text{and} \quad
    \eta^{j} \leq \sum_{\upsilon \in [\Upsilon]} \eta_{\upsilon}^{j} \, ,
\end{equation*}
we can obtain
\begin{align*}
    \min_{j* \in [M]}\tau \cdot \sum_{\upsilon \in [\Upsilon]}F_{\upsilon}(j^{*}) 
    \leq& \min_{j \in [M]}\sum_{\upsilon \in [\Upsilon]} \Evictions_{\upsilon}(j) + k\Upsilon \\
    \leq& \min_{j \in [M]} \sum_{\upsilon \in [\Upsilon]} \Big( \big\vert \lbrace t \in [t_{1}^{\upsilon}, t_{\tau}^{\upsilon}] \vert e_{t} \neq \perp \rbrace \big\vert + 6 \cdot \eta_{\upsilon}^{j} + 5k \Big) + k\Upsilon \\
    = &6\min_{j \in [M]}\eta^{j} + \mathtt{OPT} + 6k \Upsilon \, .
\end{align*}
Still by \Cref{lemma_cost_function_proxy}, we have
\begin{equation*}
    \cost({\SightlessChasing}) = \sum_{\upsilon \in [\Upsilon]}\Evictions_{\upsilon}(j_{t_{1}}^{\upsilon}) \leq \tau \cdot \sum_{\upsilon \in [\Upsilon]}F_{\upsilon}(j_{t_{1}}^{\upsilon}) \, .
\end{equation*}
Then this proof is completed.
\end{proof}

By assumption, we have $\min_{j \in [M]}\eta^{j} \in o(T)$. Therefore, \Cref{theorem_regret_multi_oracles} implies that {\SightlessChasing} has a vanishing regret.

Notice that the result in this section cannot be obtained by using the results in \cite{Emek2020stateful} directly, because the algorithms proposed in \cite{Emek2020stateful} requires to know the cache profile of the algorithm {\Sim} that follows each {\expert}, which is unavailable in the bandit access model of the NAT \expert{} setting.


\subsection{Full Information Access Model}
\label{section:full_info_oracle}

It is straightforward that {\SightlessChasing} can also guarantee the $O(kT^{2/3}\sqrt{M} + \eta^{\min})$-regret for the full information access model. 
However, it turns out that by exploiting the fact that all the prediction sequences $\big\lbrace \lbrace \predictJ{t}{j} \rbrace_{t \in [T] } \big\rbrace_{j \in [M]}$ are revealed at the beginning, we can get a better upper bound on the regret.

In particular, given the prediction sequences, one can simulate algorithm {\Sim} with using $\predictJ{}{j} = \lbrace \predictJ{t}{j} \rbrace_{t \in [T]}$ for every $j \in [M]$. For each round $t$, let the cache profile and the evicted page of {\Sim} that uses $\predictJ{}{j}$ be $\hat{C}_{t}^{j}$ and $\hat{e}_{t}^{j}$, respectively. Then for every $j \in [M]$, $\hat{C}_{t}^{j}$ and the cumulative loss $\vert \lbrace t' < t \mid \hat{e}_{t'}^{j} \neq \perp \rbrace \vert$ can be computed online through using $\lbrace \sigma_{t'}, \predictJ{t'}{j} \rbrace_{t' < t}$. This observation makes the following result applicable.

\begin{theorem}[\cite{Blum2000online}] \label{theorem_Blum_online_learning}
Given any $Q$ algorithms for the online paging problem with cache size $k$, if at each round $t \in [T]$, the cache profile and the cumulative loss of every algorithm $q \in [Q]$ are known, there exists a random online procedure that combines these $Q$ algorithms and guarantees that on any request sequence $\sigma$, the expected cost is at most
\begin{equation*}
    (1 + 2\epsilon)\cost^{*} + \Big(\frac{1}{\epsilon} + \frac{7}{6} \Big)k \cdot \log(Q)
\end{equation*}
for any $\epsilon < \frac{1}{4}$, where $\cost^{*}$ is the cost incurred by the best of the $Q$ algorithms on $\sigma$.
\end{theorem}


\begin{theorem} \label{theorem_regret_multi_fA}
For the full-information access model of the NAT \expert{} setting, there exists an online paging algorithm having an $O(\sqrt{kT\log M} + \eta^{\min} + k)$-regret.
\end{theorem}

\begin{proof}
Denote the algorithm {\Sim} that follows the predictions of each {\expert} $j \in [M]$ by $\Sim^{j}$. Then $\lbrace \Sim^{j} \rbrace_{j \in [M]}$ can be viewed as $M$ algorithms and mixed with the procedure described in \Cref{theorem_Blum_online_learning}. 
By choosing $\epsilon \in \Theta\Big(\sqrt{\frac{k\log M}{T}}\Big)$, the expected cost incurred by combining these $M$ algorithms is bounded by
\begin{equation*}
    \min_{j \in [M]}\cost\big( \Sim^{j} \big) + O\left( \sqrt{k T \cdot \log M} \right) \, .
\end{equation*}
Then this theorem follows from \Cref{theorem_single_expert_bound}.
\end{proof}

\clearpage
\appendix

\begin{figure}[!t]
\centering
\LARGE{APPENDIX}
\end{figure}

\section{Regret Lower Bound}
\label{section:impossibility}

By proving \Cref{theorem:logarithmic-gap-with-no-additional-info}, we show in this section that for any online paging algorithm \Alg{}, the regret of \Alg{} is low-bounded by $\Omega\Big( \frac{T}{k} \Big)$
when no additional information is available.

\begin{proof}[Proof of \Cref{theorem:logarithmic-gap-with-no-additional-info}]
Consider the random instance $\sigma$ described in the theorem's statement.
By the linearity of expectation, the expected cost of any online paging algorithm \Alg{} satisfies
$\mathbb{E}[\cost(\Alg)] \in O\Big( \frac{T}{k + 1} \Big)$.
To analyze $\OPT$, we partition the $T$ rounds into phases so that each phase is a minimal time interval during which every page is requested at least once.
By definition, if \FitF{} experiences a cache miss in round $t$, then its cache configuration at time
$t + 1$
contains all the pages that will be requested until the end of the current phase, possibly, with the exception of the last one.
Therefore, \FitF{} suffers at most $2$ cache misses per phase.
Since the phase lengths $L$ are i.i.d.~random variables, it can be inferred from renewal theory that $\OPT$ is up-bounded by
$O\Big( \frac{T}{\mathbb{E}[L]} \Big)$.
The proof is completed by observing that each phase corresponds to an instance of the coupon collector problem with
$n = k + 1$
coupons, hence
$\mathbb{E}[L] = \Theta (n \log n) = \Theta (k \log k)$.
\end{proof}

\section{Explicit {\Experts}}
\label{section:multiple-experts-fA}

Our goal in this section is to quantify the difference in the numbers of prediction errors between the explicit {\expert} setting and the NAT \expert{} setting.

\begin{lemma}\label{lemma_errors_transition}
If a NAT {\expert} $j$ is consistent with a page sequence $\pi$, then
\[
\eta_{e}(\pi) - n \leq \eta_{N}^{j} \leq 2\eta_{e}(\pi) \, .
\]
\end{lemma}

\begin{proof}
For each round $t$ with $\predictJ{t}{j} \neq \NextTime{t}{\sigma_{t}}$, let $t' = \min\lbrace \predictJ{t}{j}, \NextTime{t}{\sigma_{t}} \rbrace$. We first prove that $\sigma_{t'} \neq \pi_{t'}$. Suppose that on the contrary $\sigma_{t'} = \pi_{t'}$. Consider the case where $t' = \predictJ{t}{j} < \NextTime{t}{\sigma_{t}}$. Since $\predictJ{}{j} = \lbrace \predictJ{t}{j} \rbrace_{t \in [T]}$ is consistent with $\pi$, we have $\sigma_{t'} = \pi_{t'} = \pi_{\predictJ{t}{j}} = \sigma_{t}$, which leads to a contradiction that $\NextTime{t}{\sigma_{t}} \leq t' < \NextTime{t}{\sigma_{t}}$. It can be proved in a similar way that $\sigma_{t'} = \pi_{t'}$ also results in a contradiction when $t' = \NextTime{t}{\sigma_{t}} < \predictJ{t}{j}$. Therefore, we can always map a round $t$ satisfying $\predictJ{t}{j} \neq \NextTime{t}{\sigma_{t}}$ to a round $t'$ that satisfies  $\sigma_{t'} \neq \pi_{t'}$. 

Let $t$, $t'$ be a pair of rounds with $\predictJ{t}{j} < \NextTime{t}{\sigma_{t}}$ and $t' = \predictJ{t}{j}$. Then for any round $t'' \in (t, t')$ with $\predictJ{t''}{j} < \NextTime{t''}{\sigma_{t''}}$, it can be proved that $t' \neq \predictJ{t''}{j}$. Suppose that on the contrary that $t' = \predictJ{t''}{j}$. Since $\predictJ{}{j}$ is consistent with $\pi$, we have $\sigma_{t} = \pi_{t'} = \pi_{t''}$. In such a case, $\predictJ{t}{j} = t'' < t'$, which results in a contradiction. Similarly, it can be proved that if $\predictJ{t}{j} > \NextTime{t}{\sigma_{t}}$ and $t' = \NextTime{t}{\sigma_{t}}$, then for any round $t'' \in (t, t')$ with  $\NextTime{t''}{\sigma_{t''}} < \predictJ{t''}{j}$, $t' \neq \NextTime{t''}{\sigma_{t''}}$.

Therefore, we can map each round $t$ with $\predictJ{t}{j} \neq \NextTime{t}{\sigma_{t}}$ to a round $t'$ with $\sigma_{t'} \neq \pi_{t'}$ such that at most two such rounds $t$ are mapped to the same $t'$. This proves the right part of our claim that $\eta_{N}^{j} \leq 2 \eta_{e}(\pi)$.

Let $t$ be a round such that $\pi_{t} \neq \sigma_{t}$. If there is no such round $t' \in [t - 1]$ mapped to $t$ in the way we described above, then it can be proved that for any $t' \in [t - 1]$, $\sigma_{t'} \neq \sigma_{t}$. There are at most $n$ rounds that can satisfy this condition. Therefore, we have $\eta_{e}(\pi) - n \leq \eta_{N}^{j}$. This finishes the proof.
\end{proof}


\clearpage
\bibliographystyle{alpha}
\bibliography{references}

\end{document}